\documentclass[format=acmsmall, review=false, nonacm]{acmart}

\usepackage{amsbsy}

\usepackage{booktabs} 
\usepackage[ruled]{algorithm2e} 

\SetAlFnt{\small}
\SetAlCapFnt{\small}
\SetAlCapNameFnt{\small}
\SetAlCapHSkip{0pt}
\IncMargin{-\parindent}

\usepackage{bm,array,xspace,makeidx,verbatim,epsf,psfrag,calc,fancybox}
\usepackage{color,ifthen,enumerate,float}
\usepackage{mathrsfs}
\usepackage{epstopdf}
\usepackage{nicefrac}
\usepackage{pstool}
\usepackage{subcaption}

\usepackage{cases}

\usepackage{makecell}
\usepackage{booktabs}
\usepackage{multirow}
\usepackage{array}
\makeatletter
\newcommand{\thickhline}{%
    \noalign {\ifnum 0=`}\fi \hrule height 1pt
    \futurelet \reserved@a \@xhline
}
\newcolumntype{"}{@{\hskip\tabcolsep\vrule width 1pt\hskip\tabcolsep}}
\makeatother

\renewcommand{\a}{\alpha}
\renewcommand{\b}{\beta}
\renewcommand{\t}{\theta}

\newcommand{\dif}{\mathrm{d}}

\newtheorem{assumption}{Assumption}

\newcommand{\erf}{\mathrm{\,erf\,}}

\newcommand{\cc}[1]{{\color{blue}#1}}

\setcitestyle{authoryear}

\title[Repositioning, Ride-matching, and Abandonment in On-demand Ride-hailing Platforms]{Repositioning, Ride-matching, and Abandonment in On-demand Ride-hailing Platforms: A Mean Field Game Approach}

\author{Yunpeng Li}
\affiliation{
  \institution{The Chinese University of Hong Kong, Shenzhen}
  \country{China}}
\email{liyunpeng@cuhk.edu.cn}

\author{Antonis Dimakis}
\affiliation{
  \institution{Athens University of Economics and Business}
  \country{Greece}}
\email{dimakis@aueb.gr}

\author{Costas A. Courcoubetis}
\affiliation{
  \institution{The Chinese University of Hong Kong, Shenzhen}
\country{China}  }
\email{costas@cuhk.edu.cn}

\begin{abstract}
The on-demand ride-hailing industry has experienced rapid growth, transforming transportation norms worldwide. Despite improvements in efficiency over traditional taxi services, significant challenges remain, including drivers' strategic repositioning behavior, customer abandonment, and inefficiencies in dispatch algorithms. To address these issues, we introduce a comprehensive mean field game model that systematically analyzes the dynamics of ride-hailing platforms by incorporating driver repositioning across multiple regions, customer abandonment behavior, and platform dispatch algorithms. Using this framework, we identify all possible mean field equilibria as the Karush-Kuhn-Tucker (KKT) points of an associated optimization problem. Our analysis reveals the emergence of multiple equilibria, some of them characterized by drivers pursuing distant requests, i.e., exhibiting `Wild Goose Chase' behavior, leading to suboptimal system performance.
To mitigate these inefficiencies, we propose a novel two-matching-radius nearest-neighbor dispatch algorithm that eliminates undesirable equilibria and ensures a unique mean field equilibrium for multi-region systems. The algorithm dynamically adjusts matching radii based on driver supply rates, optimizing pick-up times and waiting times for drivers while maximizing request completion rates. Numerical experiments and simulation results show that our proposed algorithm reduces customer abandonment, minimizes waiting times for both customers and drivers, and improves overall platform efficiency.

\end{abstract}

\begin{document}

\begin{titlepage}

\maketitle

\makeatletter \gdef\@ACM@checkaffil{} \makeatother 


\end{titlepage}

\section{Introduction}\label{sec:introduction}
The on-demand ride-hailing industry has undergone significant expansion in recent years. Uber’s span to over 700 cities worldwide and the pervasive presence of ride-hailing services in China signify a remarkable shift in transportation norms for people. Although ride-hailing platforms improve efficiency over traditional taxi services, they still encounter challenges and opportunities for enhancement.

One of the primary challenges lies in the complex nature of ride-hailing platforms, which are composed of a large number of agents—primarily drivers—each exhibiting strategic behavior that can significantly impact the system’s overall performance. Specifically, drivers’ repositioning strategies, driven by their anticipation of future demand and competition, play a crucial role in determining the efficiency and equilibrium of the platform. Modeling and analyzing these strategic interactions require sophisticated tools that can capture the dynamic and decentralized decision-making process of drivers.

Another significant challenge is the propensity of customers to abandon the service if they are not matched with a vehicle in a timely manner. This customer abandonment behavior not only affects the user experience but also has implications for the platform’s revenue and market reputation. Addressing this challenge necessitates a deeper understanding of customer preferences and the development of mechanisms that can balance supply and demand more effectively, ensuring that customer needs are met promptly and efficiently.

Finally, the design of efficient dispatch algorithms is paramount to preventing Wild Goose Chase (WGC) equilibria  \cite{castillo2023}, where drivers chase distant requests inefficiently, leading to increased empty driving time and decreased overall system efficiency. Traditional matching algorithms may induce non-monotonic sojourn times in the driver supply rate, resulting in multiple equilibria and suboptimal outcomes. Innovating in this area requires the development of new algorithms that can dynamically adjust to real-time conditions and optimize the matching process to minimize inefficiencies and maximize service quality.

To model the strategic behavior among a large number of agents (drivers), our study employs the framework of mean field games. This approach simplifies the complex interactions between numerous agents into a representative agent interacting with the mean field, thereby significantly reducing computational complexity. Moreover, when the number of agents is large, this approximation maintains a high degree of accuracy. Given the intricate operational and strategical details involving the platform, passengers, and drivers, constructing an exact model is exceedingly difficult, and precise analysis of such a complex model is nearly impossible. Therefore, we focus on a fundamental mean field game model that captures the core issues.
Adopting an agent-based modeling perspective, we consider a closed queueing network comprising a large number of drivers, where each driver’s sojourn time follows a general distribution. The game model we propose is a continuous-time mean field game with finite state-action space where each driver's decision model is a  semi-Markov decision model (semi-MDP). 
A key idea for analyzing our mean field game is inspired by the analysis of closed queueing networks in \cite{kelly1989}, where the mean sojourn time of a job at a queue is expressed as a function of the queue's throughput. In the context of ride-hailing, this throughput corresponds to  the driver supply rate, defined as the number of available ride-hailing vehicles per unit time within a specified region. Leveraging this insight, we transform the problem of finding mean field game equilibria into identifying the Karush-Kuhn-Tucker (KKT) points of a constrained optimization problem. This pivotal approach, originally applied to mean field games in \cite{cd23}, is further generalized and extended in our work.

Employing the aforementioned mean field game model,  this study examines the impact of customer abandonment behavior on platform dispatch algorithms. Specifically, we delve into the classical nearest-neighbor dispatch algorithm, which pairs the closest available vehicle with a customer in need. Our analysis reveals the potential for WGC. This manifests as the existence of multiple equilibria within the system, one of them being an equilibrium with longer pick-up times compared to the other equilibria. In this equilibrium, distant vehicles are matched with passengers, leading to substantial efficiency losses in the system, which theoretically could become infinitely large. Meanwhile, our analysis reveals that multiple equilibria emerge because the mean pick-up time is a non-monotone function of the driver supply rate in the given region.

To prevent the emergence of such long pick-up time equilibria, we investigate the effects of controlling the matching radius -- the maximum distance within which a match is allowed. Our findings indicate that the careful selection of the matching radius is crucial, as an improper choice may result in multiple equilibria exhibiting inefficiency.
In light of the difficulties associated with analytically determining the optimal matching radius, our study proposes the Two-Matching-Radius Nearest-Neighbor Dispatch Algorithm, which can be viewed as a generalization of the traditional nearest-neighbor dispatch algorithm. This novel approach designs distinct matching radii for two scenarios: the arrival of new customers and the arrival of new available vehicles. Rather than seeking static optimal values for a fixed system, we derive an explicit relationship between the optimal matching radii and the driver supply rate.
Furthermore, we demonstrate that the nearest-neighbor dispatch algorithm utilizing the optimal adjustment of matching radii results in a sojourn time that is a monotone function of the driver supply rate. This allows us to prove the uniqueness of the mean field equilibrium, eliminating any other inefficient equilibrium. By minimizing the sojourn time for each given driver supply rate through our optimal matching radius, our solution, besides eradicating undesirable equilibria, also maximizes the system’s throughput, thereby optimizing overall system efficiency for the ride-hailing platform.

In summary, our main contributions are:
    \begin{itemize}
  \item  We introduce a mean field game model to holistically analyze the ride-hailing systems, incorporating key aspects such as driver's repositioning across multiple regions, customer's abandonment behavior, and platform's dispatch algorithms. This unified model allows for the study of their joint effects on system performance.
  
  \item In this general multi-region mean field game framework, we identify all possible equilibria as the Karush-Kuhn-Tucker (KKT) points of an associated optimization problem. Our model reveals the emergence of multiple equilibria, some of them exhibiting large pick-up times -- the `Wild Goose Chase' behavior in ride-hailing systems extensively discussed in the literature \cite{castillo2023}. In such an equilibrium, drivers are matched with distant requests and consequently the system efficiency is compromised.

  \item  We relate for the first time the choice of the matching radius with generating multiple system-wide equilibria. We propose a simple yet innovative `two-matching-radius nearest-neighbor' dispatch algorithm. This algorithm ensures a unique mean field equilibrium for the multi-region model. This implies a determinate behaviour for the whole system by eliminating multiple equilibria, some of them inefficient.
  
  \item For a single region, our proposed dispatch algorithm minimizes the total pick-up time and waiting time for drivers, thereby maximizing the request completion rate. Numerical experiments and simulation results further demonstrate that the algorithm  minimizes customer's waiting time, reduces customer abandonment, and enhances overall efficiency of the ride-hailing system.

     \end{itemize}

\subsection{Literature Review}
Our work is related to the following streams of literature.

\textbf{Repositioning in Ride-hailing Platforms.} 
\cite{dairouting} studied the repositioning of empty
cars in a spatial ride-sharing network where a platform controls the repositioning of all vehicles. \cite{afeche2023} studied  strategical repositioning in ride-hailing networks, where drivers are allowed to make repositioning decisions to maximize their earnings. \cite{benjaafar2021a} and \cite{gao2024} considered a combination of decentralized and centralized repositioning to investigate the impact of autonomous vehicles on ride-hailing platforms. Our work is closely related to \cite{gao2024} in that we similarly consider a mean field game model and characterize the mean field equilibrium by an optimization problem. In the above literature, the matching procedure is simplified assuming a fluid model. It  explicitly or implicitly assumes that rider requests are lost if not matched instantly with an available driver. In our model, both riders and drivers wait, and waiting riders abandon the system at a given rate. Our model includes a  detailed matching procedure between riders and drivers, which affects  the strategic repositioning behavior of drivers.
 
 \textbf{Ride-matching and Abandonment.} Another stream of research  close to our work is the study of the performance of matching  in ride-hailing systems. \cite{Besbes2022} considers 
 nearest-neighbor dispatch and uses a spatial queueing model to derive the service rate as  a function of pick time and en-route time. \cite{wang2022}  extends the spatial model in \cite{Besbes2022} to include admission control for matching based on a pickup-time threshold. In addition, abandonment and cancellation on the customer's side are investigated in \cite{wang2022}. Our model inherits the customer's abandonment behavior model from \cite{wang2022}, and considers a two-sided queueing model with an open customer-side and a closed driver-side queueing network. \cite{castillo2023} investigated  the `Wild Goose Chase' caused by long pick-up times. \cite{castillo2023} suggested that setting a maximum matching radius could improve the system performance. \cite{Yang2020} studied the impact of matching radius and matching time interval on the system performance. \cite{wang2022}   and \cite{Yang2020} both considered batch matching while \cite{wang2022} proposed to manipulate the matching radius and matching time interval indirectly through a pickup-time threshold. 
Unlike \cite{wang2022} and \cite{Yang2020}, we focus on instant matching using a nearest-neighbor dispatch approach. Specifically, an arriving customer is immediately matched with the closest available idle driver if one is present, and an arriving driver is immediately matched with the closest waiting customer if one is available.  Thus, our work views matching radius  as the key decision variable of the platform. Our research contributes to this stream of research by introducing a novel two-matching-radius nearest-neighbor dispatch algorithm. The dynamical optimal adjustment of the matching radii proposed in this work offers a solution that mitigates the WGC problem.

  \textbf{Mean Field Game.} Another area closely related to our work is mean-field games, which investigate the repeated interactions among a large number of strategic agents (\cite{jovanovic,huang,lions,hopenhayn,adlakha}). Our mean-field game model is of finite type, featuring finite state and action spaces. Existing finite mean-field game models in the literature are typically based on either discrete-time Markov Decision Processes (MDPs) \cite{guo2022} or continuous-time MDPs with exponentially distributed sojourn times \cite{neumann2020}. In our model, we model each driver’s decision-making by continuous-time semi-Markov decision processes (SMDPs) where sojourn times follow general distributions. In fact, our model incorporates a closed queueing network, where drivers can be viewed as jobs navigating through the network.    
 To analyze the closed queueing network, we adopt the approach proposed in \cite{kelly}, which represents mean sojourn times as functions of throughputs (corresponding to driver supply rates in the ride-hailing model). This approach enables us to determine the throughputs within the closed queueing network. Building on this foundation, we utilize the methodology outlined in \cite{cd23} to formulate an optimization problem—specifically, the minimization of a potential function—to compute mean-field equilibria. However, the complexity of the closed spatial queueing network inherent in ride-hailing systems often results in sojourn times that are non-monotone functions of throughput, thereby violating the monotonicity assumptions presented in \cite{kelly1989} and \cite{cd23}. To address this challenge, we extend their methodologies to accommodate non-monotone cases. Furthermore, by optimally designing the dispatch algorithm and matching radius, we can ensure that sojourn times exhibit monotonicity, thereby satisfying the assumptions required in \cite{kelly1989} and \cite{cd23}.

\subsection{Paper Organization}
In Section 2, we describe each driver's  continuous-time semi-Markov Decision Process model and formulate a mean field game model. In Section 3, we investigate the ride-matching algorithm under abandonment behavior, and derive the sojourn times as exact functions of  driver supply rates. 
In Section 4, we provide the
optimal matching radius calculation that eliminates multiple equilibria.
In Section 5, we conduct simulations to evaluate the performance of our proposed dispatch algorithm. We conclude the paper in Section 6. 

\section{Mean Field Game Model}
\label{sec:model}

This section introduces a mean-field game model that focuses on modeling drivers' strategic repositioning behavior across multiple regions. The details concerning customer behavior and platform strategy will be discussed in subsequent sections.

\subsection{A Semi-Markov Decision Process}
We consider a \emph{closed network} consisting of a total mass $m$ of drivers operating across $N$ regions, labeled $1, 2, \dots, N$. These regions represent distinct geographic areas within a city. Driver behavior is modeled as a continuous-time semi-Markov decision process (semi-MDP) with state space $S = {1, 2, \dots, N}$ the set of regions. The state of the MDP represents the region in which a driver just deposited a customer or has arrived without one.
The action space, denoted by $A = {1, 2, \dots, N}$, represents the driver’s choice of seeking their next customer. Since both states and actions are indexed by regions, we have $S = A$.

For a given state-action pair $(i, k) \in S \times A$, the choice $k = i$ indicates that the driver remains in region $i$ until a next customer is found locally, whereas $k \ne i$ corresponds to a `repositioning' decision: the driver travels empty to region $k$, where they will make a new decision upon arrival.

\textbf{Transition Probability.} The transition probabilities are encapsulated in the tensor 
\(P=(P^k_{ij})_{ k\in A\atop i,j\in S}\), where $P^k_{ij}$ denotes the probability of transitioning from state $i$  to state $j$ under action $k$. Let $(q_{ij})_{i,j\in S}$  denote the `demand matrix', where $q_{ij}$ is the probability that a passenger in region $i$ has a destination in region $j$. Using the demand matrix, we are able to define the transition matrix of the semi-MDP as follows.

\begin{itemize}
    \item $P_{ij}^k, k\not=i$. The action of the driver is to reposition (without customer) to region $k\not=i$. Hence, $P_{ij}^k=\delta_{jk}$, where $\delta_{jk}=1$ if $j=k$ and $0$ otherwise.
    \item $P_{ij}^k, k=i$. The action of the driver is to seek a customer in the same region he  currently is, and this customer will dictate his next region where he will be free. Hence, $P_{ij}^i =q_{ij}$.
\end{itemize}


\textbf{Customer Waiting.}
For each region $i$, customers arrive at a  rate of $b_i>0$ and wait in queues if they are not matched with drivers upon arrival.  Waiting customers abandon the system at a rate $\theta_i$ specific to each region. A detailed analysis of customers’ abandonment behavior is provided in Section \ref{section:analysisofequilibria} and is not required to derive the results in this section; instead, we only rely on the general structural properties of the drivers' sojourn times to characterize the equilibria of the game.

\textbf{Sojourn Times.} For a state-action pair $(i,k)$, it takes a certain amount of time for a driver to complete action  $k$
 when they are in state/location $i$ (to serve a customer in region $i$ when $k=i$ or reposition to $k$ if $k\not=i$ ). This is also referred to as  the \emph{sojourn time} for the state-action pair $(i,k)$.  Let $T_{ik}$ denote the mean time taken for a driver to finish action $k$ from state/location $i$.
The sojourn time is a result of interactions across the platform, customers, and drivers, encompassing elements such as the driver's wait time until a match, the duration of pick-ups, and travel times between locations. Factors influencing the determination of sojourn time include customer behavior, platform matching policies, and drivers' strategies.
Under appropriate assumptions and within specific customer behavior models and platform matching policies, we will show that the mean sojourn time can be calculated based on the distribution of driver masses over the state-action space and their respective mass rates.

\textbf{Steady State Analysis.} We consider a non-atomic game model where a representative driver makes decisions against the mass behavior of all other drivers, characterized by the mass distribution over $S\times A$. This representative driver is negligible and  exerts no influence on the system, including the mass distribution and corresponding mass rate.  This approach is reasonable given the enormous number of drivers in the system, their anonymity, symmetry, and the sole differentiation based on their locations.
 We focus on \emph{stationary equilibria} and study \emph{stationary systems}, which is sensible since each driver perceives the system to be in a steady state, enabling observation of mass distributions and estimation of sojourn times effectively. Our approach involves a mean field based analysis that operates with the averages of the sojourn times and associated metrics, a methodology commonly employed in related literature \cite{dairouting,Besbes2022,wang2022,cd23}.

 For each state-action pair $(i,k)\in S\times A$, let $\mu_{ik}$ represent the mass occupying that specific state-action pair. 
For any total driver mass $m>0$, let $\mathcal{U}_m(S\times A)=\{\boldsymbol{\mu}\in\mathbb{R}_+^{N^2}\mid \sum_{(i,k)\in S\times A}\mu_{ik}$$=m\}$, i.e., the set of all possible mass distributions over $S\times A$. Let $x_{ik}$ denote the long-run average  mass rate for each state-action pair $(i,k)$ and $x\in\mathbb{R}_+^{N^2}$ denote the average mass rate vector\footnote{In other words, $x_{ik}$ drivers, on average, arrive in region $i$ and perform action $k$ within one unit of time.
}.  We will establish the relationship between $x$ and $\mu$  in the following.
It should be noted that $x_{ii}$ denotes the average mass rate of available drivers willing to provide services within region $i$, in other words, the average trip supply rate in region $i$.
 It is clear that in any stable equilibrium the average  trip supply rate is lower than the customer arrival rate, ensuring $x_{ii}\le b_i$.

  A system characterized by  mass distribution $\mu$ and corresponding mass rate vector $x$ is considered stationary when   $x$ satisfies the following balance equation,
  \begin{equation}
      \label{eq:stationary}
      \sum_{k\in A} x_{ik}=\sum_{{k\in S},k\not=i}x_{ki}+\sum_{k\in S}x_{kk}q_{ki},\>  \forall i\in S\,,
  \end{equation}
which is interpreted as the mass flow out of state/location $i$ (on the left-hand side) should equal the  mass flow into state/location $i$ (on the right-hand side), where the latter is composed of the flow of empty cars for repositioning and the flow of cars carrying passengers.
Henceforth, unless explicitly stated otherwise, we focus on stationary systems.

We now determine the sojourn times in a stationary system, where all quantities refer to their expected (mean) values.  
Let $t_{ij}$ denote the expected trip duration from region $i$ to region $j$; in particular, $t_{ii}$ represents the expected duration of a trip within region $i$.

Consider any action $(i, k) \in S \times A$. If $k \ne i$, the vehicle performs an empty repositioning trip from region $i$ to region $k$, and the sojourn time is simply the travel time: $T_{ik} = t_{ik}$.

If $k = i$, the vehicle remains in region $i$ to serve a customer, and the sojourn time consists of three components:
\begin{enumerate}
    \item the expected waiting time until the driver is matched to a customer, denoted $w_i^d$;
    \item the expected pick-up time, denoted $\tau_i$; and
    \item the expected trip time, denoted $t_i$.
\end{enumerate}

The trip time $t_i$ is computed as the average over all possible customer destinations, weighted by the demand probabilities:
\[
t_i = \sum_{j \in S} t_{ij} q_{ij}.
\]

Note that both the driver’s waiting time $w_i^d$ and the pick-up time $\tau_i$ are affected by the customer arrival process and the platform’s matching policy.
The relationships between the driver's average waiting time, pick-up time, and mass rate can be established in each region based on the customer behavior model and the platform's matching policy. We will investigate these relationships in detail in subsequent sections. However, in this section, rather than deriving their exact forms, we focus on their \textbf{structural properties}, which are sufficient for characterizing and proving general equilibrium properties. These properties should hold in general ridehailing scenarios when customers are impatient, and we will formally establish them within our modeling framework in subsequent sections.


\begin{assumption}\label{assum:function_of_x}
For a given vector of customer arrival rates and platform matching policy, when the system is in a stationary state, the following holds for every \( i \in S \):
\begin{enumerate}
    \item The driver's mean waiting time \( w_i \) and mean pick-up time \( \tau_i \) are positive and can be written as functions of the mass rate of available drivers \( x_{ii} \). Consequently, the mean total sojourn time \( T_{ii} \) can also be expressed as:
    \[
    T_{ii}(x_{ii}) = w_i^d(x_{ii}) + \tau_i(x_{ii}) + t_i,
    \]
    where \( t_i = \sum_{j \in S} t_{ij} q_{ij} \) is a constant.
    \item As \( x_{ii} \rightarrow b_i \) (the demand-supply limit), the waiting time grows unboundedly:
    \[
    w_i^d(x_{ii}) \rightarrow \infty.
    \]
\end{enumerate}
\end{assumption}

The idea of expressing \( T_{ii} \) as a function of throughput (\( x_{ii} \)) stems from closed queuing network analysis \cite{kelly1989}. In Section \ref{section:analysisofequilibria}, we will demonstrate that this assumption holds for our matching mechanism and customer abandonment model, providing explicit formulae for validation.
The condition \( w_i^d(x_{ii}) \rightarrow \infty \) as \( x_{ii} \rightarrow b_i \) arises from \textbf{customer impatience}: in order to serve all customers, given that customers abandon the service at a positive rate, we need an unbounded idle driver pool.


For notational uniformity, we write the sojourn time as \( T_{ik}(x) \) for all \( (i, k) \in S \times A \), even though it depends only on \( x_{ik} \) when \( k = i \) and remains constant otherwise.

Given Assumption \ref{assum:function_of_x}, the mass rate satisfies the following equation for a stationary system,
\begin{equation}\label{eq:little_law}
    \mu_{ik}=x_{ik}T_{ik}(x),\>\forall(i,k)\in S\times A.
\end{equation}
As the total mass is $m$, immediately, we obtain the conservation law of mass,
\begin{equation}
    \label{eq:conservation}
    \sum_{(i,k)\in S\times A}x_{ik}T_{ik}(x) =\sum_{i\in S} x_{ii}\big(w_i^d(x_{ii})+\tau_i( x_{ii})+t_i\big)+\sum_{i,j\in S\atop i\not=j}x_{ij}t_{ij} =m.
\end{equation}

\textbf{Driver's Strategy.} In this section, our focus is on the repositioning strategy of drivers. After completing a trip, drivers face a decision problem where they must choose between serving their current location or repositioning to another location that potentially offers more customers and higher income. 
We assume that every driver follows a stationary repositioning strategy, meaning that their choice of action is solely dependent on states (the current location) and is independent of time. A driver's stationary strategy is characterized by the matrix $(\pi_{ik})_{i\in S, k\in A}$. Here, $\pi_{ik}$ represents the probability that a driver selects action $k$ where $\sum_{k\in A}\pi_{ik}=1$. Given a stationary mass rate $x$, the aggregate stationary repositioning strategy, which amalgamates the strategies of all drivers, can be derived as follows,
\begin{equation}\label{eq:rate_to_strategy}
    \pi_{ik}(x)=\frac{x_{ik}}{\sum_{k\in A} x_{ik}}\,,
\end{equation}
for state $i$ with a positive denominator, while $\pi_{ik}=0$ otherwise.  In fact, the mass rates encapsulate the full information regarding drivers' strategies. As we will demonstrate later in Proposition \ref{prop:ne_mass_rate}, the optimality of drivers' strategies is derived in terms of the vector of mass rates.

\textbf{Reward Structure.} In this model, we assume\footnote{We can extend the cost model and our analysis still holds.} a normalization of zero operating costs for vehicles and zero opportunity costs for drivers. We do not address the participation problem for drivers and assume that all drivers are willing to join the platform and adhere to the matching arrangements set by the platform.  We assume that rewards are continuously generated at a compensatory rate of $c_i>0$ solely during the time when a driver is transporting a customer from the customer's pick-up location to their destination. Consequently, with each trip originating from region $i$, a driver can, on average, earn $c_it_i$. Given a stationary system, the average reward per driver per unit of time, denoted by $\Phi$, is a function of the stationary mass rate vector $x$ given by
\begin{equation}\label{eq:average_reward}
    \Phi(x) =  \frac{1}{m}\sum_{i\in S}x_{ii}c_it_i.
\end{equation}
This study focuses on the semi-Markov decision process with an average reward criterion where each driver aims to maximize their individual average reward over an infinite time horizon. Let  $(\tilde{\pi}_{ik})_{i\in S, k\in A}$ be a reposition strategy for a representative driver.   Let $\tilde{x}=(\tilde{x}_{ik})_{i\in S, k\in A}\in \mathbb{R}_+^{ N^2 }$ be the representative driver's mass rate where  for any $i$ and $k$, $\tilde{x}_{ik}$ represents the average number of times that this driver arrives in region $i$ and takes action $k$ within one unit of time. Then $\tilde{\pi}$ and $\tilde{x}$ are related through an equation similar to \eqref{eq:rate_to_strategy},
\begin{equation}
    \label{eq:rate_to_strategy_tilde}
    \tilde{\pi}_{ik}=\frac{\tilde{x}_{ik}}{\sum_{k\in A} \tilde{x}_{ik}}
\end{equation}
Given  the system's aggregated stationary mass rate $x$, the representative driver's $\tilde{x}$ satisfies the conservation law and balance equation as follows,
\begin{subnumcases}{}
       \sum_{i\in S} \tilde{x}_{ii}\big(w_i^d(x_{ii})+\tau_i( x_{ii})+t_i\big)+\sum_{i,j\in S\atop i\not=j}\tilde{x}_{ij}t_{ij} =1\,,\label{eq:balance_tilde}\\
        \sum_{j\in S} \tilde{x}_{ij}=\sum_{{j\in S},j\not=i}\tilde{x}_{ji}+\sum_{j\in S}\tilde{x}_{jj}q_{ji},\>  \forall i\in S\,.\label{eq:stationary_tilde}
\end{subnumcases}
The representative driver's objective is to choose a reposition strategy  
$\tilde{\pi}$ to maximize his own average reward $\sum_{i\in S}\tilde{x}_{ii}c_it_i$.

In summary, we have defined a non-atomic game, denoted by $\mathcal{G}$,  through the tuple \[\left \langle m,S,A,\big(P^k_{ij}\big)_{i,j\in S,k\in A},\Big(T_{ij}(\cdot)\Big)_{i,j\in S},(c_i)_{i\in A}\right \rangle.\]
The game $\mathcal{G}$ belongs to a mean field type.   A mean field game (MFG) is a game-theoretical framework that models the strategic decision-making of a large number of interacting agents. 
In MFGs, each agent interacts with the  `mean field' (in our context, the vector $\mu$ of agent mass per state-action pair), which represents the average effect of the strategies taken by  all the other agents. In our model, a key mean field term  is the  sojourn time  $\big(T_{ij}(\cdot)\big)_{i,j\in S}$, which implicitly depends on the mean field $\mu$  through the mass rate vector $x$. In this paper, we make no assumptions that the reader is familiar with  MFG theory, and our analysis is self-contained.

\subsection{Stationary Equilibrium}

We give the definition of a stationary equilibrium of the game $\mathcal{G}$.
\begin{definition}\label{def:ne}
   A stationary equilibrium of game $\mathcal{G}$ is a vector $(\mu^\dagger,x^\dagger)\in\mathcal{U}_m(S\times A)\times \mathbb{R}_+^{N^2}$ such that
   
   \begin{itemize}
       \item[(a)] $x^\dagger$ is the  mass rate under the  mass distribution $\mu^\dagger$ and it is stationary, i.e., $(\mu^\dagger,x^\dagger)$ satisfies \eqref{eq:stationary}, \eqref{eq:little_law}, and \eqref{eq:conservation}.
       \item[(b)] Given the aggregate behavior of all the drivers  characterized by $(\mu^\dagger,x^\dagger)$, the average reward per unit mass $\Phi(x^\dagger)$ in \eqref{eq:average_reward}  is the optimal average reward for a representative driver. In other words, the behavioral strategy (aggregated strategy) $\pi(x^\dagger)$ in \eqref{eq:rate_to_strategy} is the best response against $(\mu^\dagger,x^\dagger)$.
   \end{itemize}
\end{definition}

The second condition states that no driver can achieve a higher average reward by deviating from the current aggregated behavioral strategy.  A stationary equilibrium of the game can be fully characterized by the mass rate vector, as stated in the following proposition.
\begin{proposition}
\label{prop:ne_mass_rate}
    $x^\dagger$ is an equilibrium mass rate of game $\mathcal{G}$ if and only if $x^\dagger$ is an optimal solution of the following optimization problem,
    \begin{subequations}\label{eq:ne_mass_rate}
      \begin{align}
\max \quad& \sum_{i\in S}x_{ii}c_it_i\tag{\ref{eq:ne_mass_rate}}\\
\mbox{s.t. } &\sum_{i\in S}  x_{ii}\big(w_i(x^\dagger_{ii})+\tau_i(x^\dagger_{ii})+t_i\big)+\sum_{i,j\in S\atop i\not=j}x_{ij}t_{ij} =m,\label{conservation}
\\&  \sum_{j\in S}x_{ij}=\sum_{j\in S,j\not=i}x_{ji}+\sum_{j\in S}q_{ji}x_{jj},\>  i\in S, \label{balance_ne}\\
\mbox{over } & x=(x_{ij})_{i,j\in S}\in\mathbb{R}_+^{ N^2 }.
\end{align}
\end{subequations}
\end{proposition}
\begin{proof}
   Consider a  mass rate $x^\dagger\in\mathbb{R}_+^{ N^2 }$, we will identify the conditions for  $x^\dagger$ being  an equilibrium mass rate. Note that we can compute the corresponding distribution $\mu^\dagger$ by \eqref{eq:little_law}. Hence, the first condition for   $x^\dagger$ being  an equilibrium mass rate, namely, Definition \ref{def:ne}(a), is equivalent to that $x^\dagger$ satisfies \eqref{conservation} and \eqref{balance_ne} as per \eqref{eq:conservation} and \eqref{eq:stationary} respectively. 
    Let $\tilde{x}=(\tilde{x}_{ik})_{i\in S, k\in A}\in \mathbb{R}_+^{ N^2 }$ be the representative driver's mass rate.
    Since the representative driver's strategy is encoded in  $\tilde{x}$ through \eqref{eq:rate_to_strategy_tilde}, he seeks to manipulate $\tilde{x}$ to maximize his average reward, i.e. $\sum_{i\in S}\tilde{x}_{ii}c_it_i$. 
    A feasible  $\tilde{x}$ must satisfy \eqref{eq:stationary_tilde} in a stationary system. Moreover, as the sojourn time is determined by the aggregate mass rate $x^\dagger$, similar to \eqref{eq:balance_tilde}, a feasible $\tilde{x}$ must satisfy that 
    $$\sum_{i\in S}  \tilde{x}_{ii}\big(w_i(x^\dagger_{ii})+\tau_i(x^\dagger_{ii})+t_i\big)+\sum_{i,j\in S\atop i\not=j}\tilde{x}_{ij}t_{ij} =1.$$ Thus, the representative driver's optimization problem is, 
        \begin{subequations}\label{eq:ne_mass_rate_1}
      \begin{align}
\max \quad& \sum_{i\in S}\tilde{x}_{ii}c_it_i\tag{\ref{eq:ne_mass_rate_1}}\\
\mbox{s.t. } &\sum_{i\in S}  \tilde{x}_{ii}\big(w_i(x^\dagger_{ii})+\tau_i(x^\dagger_{ii})+t_i\big)+\sum_{i,j\in S\atop i\not=j}\tilde{x}_{ij}t_{ij} =1,\label{conservation_1}
\\&  \sum_{j\in S}\tilde{x}_{ij}=\sum_{j\in S,j\not=i}\tilde{x}_{ji}+\sum_{j\in S}q_{ji}\tilde{x}_{jj},\>  i\in S, \label{balance_ne_1}\\
\mbox{over } & \tilde{x}=(\tilde{x}_{ij})_{i,j\in S}\in\mathbb{R}_+^{ N^2 }.
\end{align}
\end{subequations}
The second condition for   $x^\dagger$ being  an equilibrium mass rate, namely, Definition \ref{def:ne}(b), is equivalent to that $\tilde{x}=\frac{1}{m}x^\dagger$ solves the optimization  problem \eqref{eq:ne_mass_rate_1}. 
Therefore, the two conditions for  $x^\dagger$ being  an equilibrium mass rate, i.e., Definition \ref{def:ne}(a) and Definition \ref{def:ne}(b) together, are equivalent to that  $x^\dagger$ solves the optimization  problem \eqref{eq:ne_mass_rate}. 
\end{proof}

Proposition \ref{prop:ne_mass_rate} provides a way to verify  whether a given  $x^\dagger$ is an equilibrium mass rate of game $\mathcal{G}$. We first compute the sojourn times $w_i(x^\dagger_{ii})$ and  $\tau_i(x^\dagger_{ii})$ to obtain the linear constraints  \eqref{conservation} and then check  if $x^\dagger$  solves the linear optimization problem \eqref{eq:ne_mass_rate}.

A key result of this paper is the establishment of a one-to-one correspondence between the stationary equilibria and the Karush-Kuhn-Tucker (KKT) points of an optimization problem. Note that for a general optimization problem, KKT points refer to those points that satisfy the KKT conditions.

\begin{theorem} \label{thm:ne_opt} Suppose  Assumption \ref{assum:function_of_x} is true,
     an $x^\dagger\in\mathbb{R}_+^{ N^2 }$ is an equilibrium mass rate  of game $\mathcal{G}$ if and only if $x^\dagger$ satisfies the KKT conditions of the optimization problem:
  \small  \begin{align}
\max \quad& m\log \bigg(\sum_{i\in S}x_{ii}c_it_i\bigg)-\sum_{i\in S}  x_{ii}t_i-\sum_{i,j\in\mathcal{S}\atop i\not=j}x_{ij}t_{ij}
-\sum_{i\in S}\int_0^{x_{ii}}w_i^d(s)\dif s-\sum_{i\in S}\int_0^{x_{ii}}\tau_i(s)\dif s \label{eq:opt_obj}\\
\mbox{s.t. } &  \sum_{j\in S} x_{ij}=\sum_{j\in S,j\not=i}x_{ji}+\sum_{j\in S}q_{ji}x_{jj},\>  i\in\mathcal{S},\nonumber \\
\nonumber
\mbox{over } & x=(x_{ij})_{i,j\in S}\in\mathbb{R}_+^{ N^2 }.
\end{align}
\end{theorem}
\begin{proof}
 Let $x^\dagger$ be an equilibrium. By Proposition \ref{prop:ne_mass_rate}, $x^\dagger$ is an optimal solution of \eqref{eq:ne_mass_rate}.
$x^\dagger$ remains optimal if the objective in \eqref{eq:ne_mass_rate} is replaced by
\begin{equation}
    m\log\bigg(\sum_{i}x_{ii}c_it_i\bigg).\label{replace}
\end{equation}
Thus, $x^\dagger$ satisfies the first order conditions:
    \begin{subnumcases}{}\label{eq:ne_foc1}
       \frac{mc_it_i}{\sum_i x^\dagger_{ii}c_it_i}-\lambda\bigg[t_i+\tau_i(x^\dagger_{ii})+w_i^d(x^\dagger_{ii})\bigg]+\sum_{j\not=i}(\nu_i-\nu_j)q_{ij}\le 0,\text{ with equality if }x^\dagger_{ii}>0,\ \forall i,\\
    \label{eq:ne_foc2}    -\lambda t_{ij}+\nu_{i}-\nu_j\le 0,\text{ with equality if } x^\dagger_{ij}>0,\ \forall i,j,i\not=j,
    \end{subnumcases}
where $\lambda\in \mathbb{R}$ is the Lagrange multiplier of the conservation law \eqref{conservation} and $\nu_i\in \mathbb{R}$ is the Lagrange multiplier of the $i$th balance equation \eqref{balance_ne}.

By multiplying both sides of \eqref{eq:ne_foc1} and \eqref{eq:ne_foc2} by the respective 
 $x^\dagger_{ij}$ for all $i$, $j$, then summing them up, we obtain:
\begin{align*}
   &\sum_{i}x^\dagger_{ii}\frac{mc_it_i}{\sum_i x^\dagger_{ii}c_it_i}-\lambda\sum_ix^\dagger_{ii}\bigg[t_i+\tau_i(x^\dagger_{ii})+w_i^d(x^\dagger_{ii})\bigg]-\lambda\sum_{j\not=i}x^\dagger_{ij}t_{ij}\nonumber\\
   &+\sum_ix^\dagger_{ii}\sum_{j\not=i}(\nu_i-\nu_j)q_{ij}+\sum_{j\not=i}(\nu_i-\nu_j)x^\dagger_{ij}=0\,.
\end{align*}
Applying the conservation law \eqref{conservation} and balance equation \eqref{balance_ne}, we obtain $m(1-\lambda)=0$, so $\lambda=1$. Therefore, the first order conditions can be rewritten as
\begin{subnumcases}{}
        \frac{mc_it_i}{\sum_ix^\dagger_{ii}c_it_i}-t_i-\tau_i(x^\dagger_{ii})-w_i^d(x^\dagger_{ii})+\sum_{j\not=i}(\nu_i-\nu_j)q_{ij}\le 0,\text{ with equality if }x^\dagger_{ii}>0,\ \forall i\label{kkt_ii}\,,\\
        -t_{ij}+\nu_{i}-\nu_j\le 0,\text{ with equality if }x^\dagger_{ij}>0,\ \forall i,j,i\not=j\,,\label{kkt_ij}
\end{subnumcases}
where $\nu_i$ is the Lagrange multiplier of the $i$th balance equation. The conditions  \eqref{kkt_ii} and \eqref{kkt_ij} are exactly the KKT conditions for the optimization problem \eqref{eq:opt_obj}.
\item Conversely, if $x^\dagger$ satisfies the KKT conditions of the optimization problem \eqref{eq:opt_obj}, i.e., \eqref{kkt_ii} and \eqref{kkt_ij}, it also satisfies the first order conditions \eqref{eq:ne_foc1} and \eqref{eq:ne_foc2} for the optimization problem \eqref{eq:ne_mass_rate} with $\lambda=1$ and objective function replaced by \eqref{replace} . 
By multiplying both sides of \eqref{kkt_ii} and \eqref{kkt_ij}  by the respective 
 $x^\dagger_{ij}$ for all $i$, $j$, then summing them up, we obtain the conservation law \eqref{conservation}. Since $x^\dagger$ also satisfies the balance equation \eqref{balance_ne}, $x^\dagger$ is feasible and hence optimal for the optimization problem \eqref{eq:ne_mass_rate}. Thus, $x^\dagger$ is an equilibrium.
\end{proof}

\begin{corollary}
    If Assumption \ref{assum:function_of_x} holds, a stationary equilibrium always exists for game $\mathcal{G}$.
\end{corollary}
\begin{proof}
By Theorem \ref{thm:ne_opt}, we only need to show that there exist KKT points. Because  optimal solutions to \eqref{eq:opt_obj} are KKT points, it suffices to show there exist such solutions. Our function is continuous and if maximized over a compact region,  it would have an optimal solution. But since our domain is not compact, we need to show that optimal values can be obtained in a compact subset of the function domain.
The current domain for variable $x$ in \eqref{eq:opt_obj}, $\mathbb{R}_+^{ N^2 }$, is open and unbounded. We need to show that the optimal value of \eqref{eq:opt_obj} can be obtained in a closed and bounded (compact) subset of $\mathbb{R}_+^{ N^2 }$. 

First, it should be noted that the optimal value of \eqref{eq:opt_obj} must be obtained away from  the zero vector, which is not a feasible solution. This is because as $||x||\rightarrow0$, the objective value tends towards negative infinity. Second,  as per Assumption  \ref{assum:function_of_x},  the mass rate $x_{ii}$ has a strict upper bound 
 $b_i$ for any $i$. Third, for any $i\not=j$, as $x_{ij}$ goes to infinity, the objective function in \eqref{eq:opt_obj} goes to negative infinity since $t_{ij}>0$. Consequently, the optimal  value of \eqref{eq:opt_obj} must be obtained when any $x_{ij}$ with $i\not=j$ is finite.
 Furthermore, the optimal value must be obtained when  $x_{ii}$ is strictly less than $ b_i$   for any $i$ since given part $(2)$ in Assumption \ref{assum:function_of_x}, the derivative of the objective function in \eqref{eq:opt_obj} with respect to $x_{ii}$ 
 \[ \frac{mc_it_i}{\sum_ix^\dagger_{ii}c_it_i}-t_i-\tau_i(x^\dagger_{ii})-w_i^d(x^\dagger_{ii})\]
goes to negative infinity as $x_{ii}$   approaches  $b_i$.
 Thus,  without loss of generality, the open domain $\mathbb{R}_+^{ N^2 }$ for $x$ can be replaced by a closed and bounded region in $\mathbb{R}_+^{ N^2 }$. Given the continuity of the objective function, the optimal problem \eqref{eq:opt_obj} must possess an optimal solution  in $\mathbb{R}_+^{ N^2 }$. 
\end{proof}

The optimization problem is convex when the (total) sojourn time $T_{ii}(x_{ii})$ increases with respect to $x_{ii}$ for every $i$.  Moreover, if this increases strictly  for every $i$, the optimization problem becomes strictly convex, leading to  uniqueness of the stationary equilibrium.

\begin{corollary}\label{coro:uniqueness}
Suppose i) Assumption \ref{assum:function_of_x} holds, and ii) the total sojourn time of a driver, $T_i(x_{ii})$,  is a strictly increasing function of  $x_{ii}$ for every $i\in S$. Then, a unique stationary equilibrium exists for the game $\mathcal{G}$.
\end{corollary}

On the other hand, if $T_i(x_{ii})$ is not an increasing function of $x_{ii}$, multiple equilibria can arise. In fact, under most dispatch algorithms, the pick-up time is not a monotonic function of $x_{ii}$. In the case of multiple equilibria, the one with the lowest $x_{ii}$ is not efficient since drivers may be dispatched to pick up customers located far away and many customer abandon the system due to longer waiting. In Section \ref{section:original_NN} we analyze the impact of driver dispatch algorithms, and propose a novel dispatch algorithm designed to ensure that $T_i(x_{ii})$ is an increasing function of $x_{ii}$. This guarantees a unique equilibrium and eliminates inefficient equilibria.


\section{Ride-Matching Algorithms and Analysis of Equilibria}\label{section:analysisofequilibria}

In this section, we present a comprehensive analysis of the platform's ride-matching algorithms and customer abandonment behavior, examining their influence on the equilibria and the resulting overall performance of the ride-hailing system.  We begin by introducing the customer abandonment behavior model in Section \ref{section:Abandonment}. Next, we present our design of the platform's matching algorithm in Section \ref{section:Matching}. Section \ref{subsection:NN} derives formulas for driver waiting and pick-up times as functions of the driver supply rate, a prerequisite (Assumption \ref{assum:function_of_x}) for applying Theorem \ref{thm:ne_opt} to identify equilibria as KKT points of an optimization problem. Building on these results,  Section \ref{section:original_NN} demonstrates that the widely adopted nearest-neighbor dispatch approach--whether implemented without a matching radius or with an arbitrarily chosen one--can lead to multiple equilibria, some of which are characterized by lengthy pick-up times. This finding emphasizes the critical need to optimize the design of the matching radius, a topic we delve into in the following section.

\begin{table}\footnotesize
    \centering 
   \caption{Notations and Units}
    \begin{tabular}{c|c|c}
    \thickhline 
 \makecell{Notation} & Physical Meaning & Suggested Unit\\
    \thickhline 
   $m$ & Overall mass density of the drivers (vehicles)  & $\text{veh}/\text{km}^2$ \\
   \hline
   $\mu_{ik}$ & Mass density of the drivers in region $i$ taking action $k$  & $\text{veh}/\text{km}^2$\\
   \hline
   $t_{ij}$ & Expected travel time from region $i$ to $j$ & min\\
   \hline
   $x_{ik}$ & \makecell{Rate of the drivers arrive in region $i$ 
   and take action $k$ per unit area }& $\text{veh}/( \text{min}\cdot \text{km}^2)$\\
    \hline
   $c_i$ & \makecell{Money paid to the driver per unit time when carrying a passenger} & $\$/\text{min}$\\
    \hline
   $b_i$ & \makecell{Mass of customers (passengers) arrive in region $i$ per unit time per unit area} & $\text{pass}/( \text{min}\cdot \text{km}^2)$\\
   \hline
   $\theta_i$ & Customer abandonment rate in region $i$ & $\text{min}^{-1}$\\
   \hline
   $R_i$ & Matching radius in region $i$ & km\\
   \hline
   $w_i^d$ & Driver's expected waiting time in region  $i$ & min \\
   \hline
     $w_i^c$ & Customer's expected waiting time in region  $i$ & min \\
   \hline
    $\tau_i$ & Expected pick-up time in region  $i$ & min \\
   \hline
    $v_i$ & Average speed of the vehicles in region  $i$ &  $\text{km}/\text{min}$  \\
    \hline
   $\mu_{i}^d$ & Mass density of the waiting drivers in region $i$  & $\text{veh}/\text{km}^2$\\
   \hline
   $\mu_{i}^c$ & Mass density of the waiting customers in region $i$ & $\text{pass}/\text{km}^2$\\
     \thickhline 
    \end{tabular}\label{table:units}
\end{table}

\textbf{Remark on Region Size and Units.} 
In Section \ref{sec:model}, we did not specify the size of each physical region because the methodology and analysis presented there are independent of the specific dimensions of individual regions. However, the size of each region becomes relevant when we introduce the matching radius in this section. Performance metrics of the ride-hailing system—such as waiting times and pick-up times—are sensitive to both the size of regions and the mass per unit area within them.
To reflect this, we now redefine the mass and rate introduced in Section \ref{sec:model} as mass per unit area (or mass density) and rate per unit area, respectively. For brevity, we may occasionally omit the term "per unit area" when the context is clear.
For clarity, Table \ref{table:units} summarizes the physical meanings and suggested units for the primary variables and parameters used in this paper, including those introduced later. These units will be employed in the numerical experiments and simulations conducted throughout the study.


\subsection{Customer: Abandonment Behavior}\label{section:Abandonment}
We introduce the the customer abandonment behavior model in this section. 
Ride-hailing platforms typically permit customers to cancel the service if they consider their waiting time to be excessively long. Our model, similar to \cite{Yang2020,wang2022}, assumes that each customer has a (random) maximum willing waiting time. If a customer is not matched with a driver within this time limit, they will exit the system. Additionally, we assume that once a match is made, customers are committed to the ride and cannot abandon the service.

In our equilibrium analysis we assume that each waiting customer in region $i$ has an exponential alarm clock with parameter $\theta_i$, that if it expires prior to the customer being matched with a driver, this customer abandons the system.
Thus, $\theta_i$ represents the impatience level of customers in region $i$: a higher value of $\theta_i$ indicates greater impatience among customers. 
In the spirit of mean-field approximation,  if the mass density of waiting customers is $\mu_i^c$, these will generate a rate of abandonment $\theta_i\mu_i^c$, where
 \begin{equation}
    \label{eq:abandon}
    \theta_i\mu_i^c+x_{ii}=b_i\,,
\end{equation}
   i.e.,  customer abandonment rate + service rate (driver supply rate) = customer arrival rate in an equilibirium state.
This equation establishes a relationship between the unknown variables  $x_{ii}$, $\mu_i^c$ and known parameters $b_i$, $\theta_i$ using mean-field approximation.


\subsection{Platform: Matching Policy}\label{section:Matching}

\begin{figure*}
    \centering
    \includegraphics[width=4in]{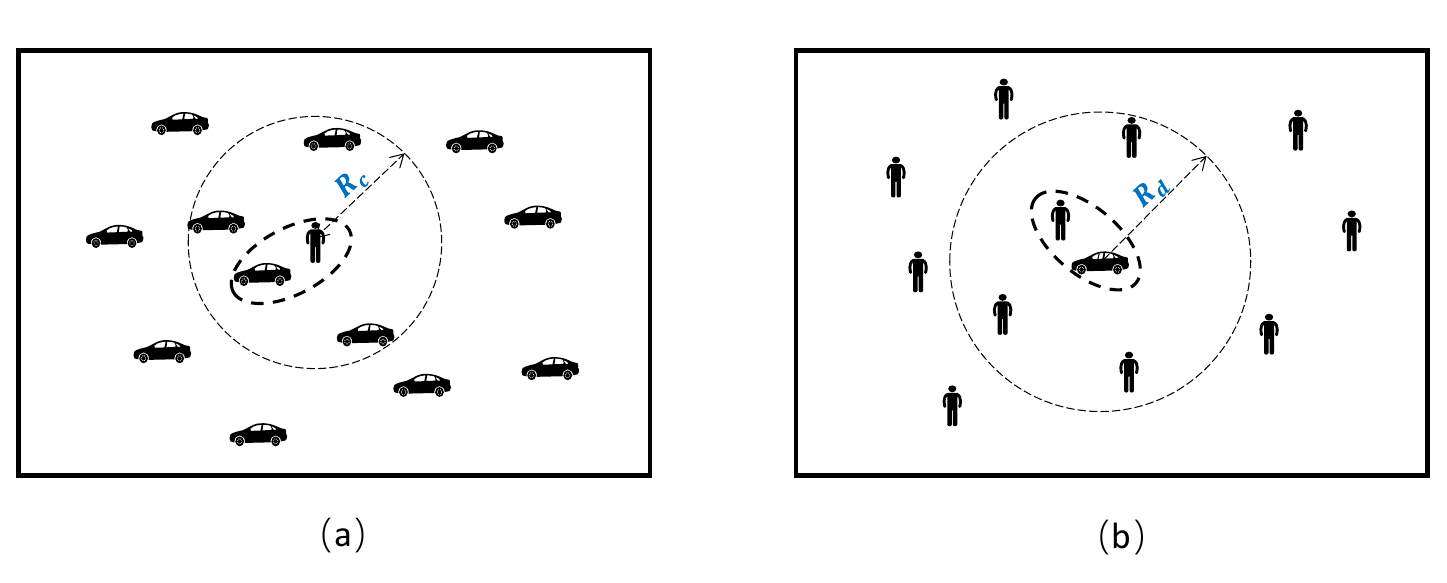}
    \caption{Two-matching-radius Nearest-neighbor Dispatch: (a). an arriving customer is immediately matched with the closest idle driver within matching radius $R_c$ and (b). an arrival driver is immediately matched with the closest waiting customer within matching radius $R_d$. }
    \label{fig_NN}
\end{figure*}

We present our design for  the platform's matching algorithm in this section. We assume the platform adopts `nearest-neighbor dispatch' (NN), which is simple and near-optimal \cite{Besbes2022}.   Recent work proposes using a matching radius to prevent drivers from being dispatched to pick up customers who are too far away, thereby improving the system performance \cite{castillo2023,Yang2020}.
Our observation is that only one matching radius lacks flexibility. A key innovation of this study is the introduction of a \textbf{two-matching-radius nearest-neighbor dispatch} algorithm:
\begin{itemize}
    \item When a new customer arrives (as in Fig. \ref{fig_NN}.(a) ), the platform assigns him to the closest available driver within the first matching radius $R_c$.
    \item When a driver becomes available (as in Fig. \ref{fig_NN}.(b) ), the platform assigns him to the closest waiting customer within a second matching radius $R_d$, which may differ from the first matching radius.
\end{itemize}
Drivers who are not immediately matched upon arrival will wait until they are paired with incoming customers. Similarly, customers who are not instantly matched upon arrival will either wait to be paired with incoming drivers or abandon the service if they become impatient.

\subsection{
Formulas for Waiting Times and Pick-up Times}\label{subsection:NN}

In this section we establish formulas for the average waiting time and pick-up time as functions of the driver supply rate, as required (Assumption \ref{assum:function_of_x}) for applying Theorem 2.3 to identify equilibria through an optimization problem, based on the customer abandonment behavior defined in Section \ref{section:Abandonment} and the platform’s dispatching algorithm defined in Section \ref{section:Matching}.

 Let  each region $i$'s area be denoted by $a_i$.  
We introduce two matching radii in region $i$, $R_{i,c},R_{i,d}\in(0,\sqrt{a_i/\pi})$, where  $\sqrt{a_i/\pi}$ is the approximate maximum radius of region $i$.   
 We denote the average vehicle speed by $v_i$. Let $\mu_i^c$ and $\mu_i^d$ denote the average densities of waiting customers and waiting drivers respectively across the region $i$. Recall that $x_{ii}$ denotes the rate of drivers that  become available and choose to serve region $i$ per unit area, and $b_i$, $\theta_i$ the arrival rate and abandonment rate of customers per unit area, respectively.
 
We assume the locations of waiting customers and waiting drivers are drawn at any moment from independent uniform distributions across region \(i\). This assumption is justified by the randomness of customers' origins and destinations, as well as the mobility of the vehicles. We can now derive the formula for the driver's average waiting time. A critical observation is that the rate at which drivers enter the waiting state is equal to the rate at which they leave it. The entry rate is exactly the rate at which drivers become idle after failing to be matched with customers, while the exit rate is the rate at which the waiting drivers are matched with arriving customers. Using this key insight, we first establish a relationship between  $\mu_i^c$ and $\mu_i^d$:
\begin{equation}
    x_{ii}\left(1-\frac{\pi R_{i,c}^2}{a_i}\right)^{\mu_i^ca_i}=b_i\left(1-\left(1-\frac{\pi R_{i,d}^2}{a_i}\right)^{\mu_i^da_i}\right)\,,\label{eq:steady_state}
\end{equation}
where the left-hand side represents the approximate average rate at which drivers enter the waiting state, while the right-hand side represents the approximate average rate at which the waiting drivers are matched with arriving customers.
With the relations \eqref{eq:abandon} and  \eqref{eq:steady_state}, as well as Little's law, $\mu_i^d=x_{ii}w_i^d$, we are able to derive the  formulas for driver's average 
 waiting time and pick-up time, whose detailed derivation are provided in Appendix \ref{appendix:proof_formulas}.
    \begin{proposition}\label{prop:formulas}
    Assume the platform adopts the two-matching-radius nearest-neighbor dispatch with radii $R_{i,c}$ and $R_{i,d}$ in each region  $i\in S$.  Then, the driver's mean waiting time and mean pickup time can be expressed as explicit functions of the driver supply rate $x_{ii}$ as follows,
       \begin{align}
       w_i^d(x_{ii})&=\frac{\log(1-\frac{x_{ii}}{b_i}(1-\frac{\pi R_{i,c}^2}{a_i})^{\frac{b_i-x_{ii}}{\theta_i}a_i})}{x_{ii}a_i\log(1-\frac{\pi R_{i,d}^2}{a_i})}\,,\label{eq:waiting_time}\\
    \tau_i(x_{ii})&=\frac{1}{v_i}\int_0^{R_{i,c}}(1-\frac{\pi r^2}{a_i})^{\mu_i^ca_i}-(1-\frac{\pi R_{i,c}^2}{a_i})^{\mu_i^ca_i}\dif r+\frac{b_i}{v_ix_{ii}}\int_0^{R_{i,d}}(1-\frac{\pi r^2}{a_i})^{\mu_i^da_i}-(1-\frac{\pi R_{i,d}^2}{a_i})^{\mu_i^da_i}\dif r.
    \label{eq:pick-up}
\end{align}
    \end{proposition}

Proposition \ref{prop:formulas} validates Assumption \ref{assum:function_of_x}. More specifically, Proposition \ref{prop:formulas} states that the driver's mean waiting time and mean pickup time can indeed be expressed as explicit functions of the driver supply rate as required by Assumption \ref{assum:function_of_x}.(1)\,. Furthermore, one can check that as $x_{ii}\rightarrow b_i$, $w_i^d(x_{ii})$ in \eqref{eq:waiting_time} goes to $\infty$, which satisfies Assumption \ref{assum:function_of_x}.(2)\,.  As a result, by Theorem \ref{thm:ne_opt}, the equilibria of game $\mathcal{G}$ under our two-matching-radius nearest-neighbor dispatch algorithm  are exactly the KKT points of the optimization problem \eqref{eq:opt_obj} with $\tau_i$ and $w_i^d$ given by \eqref{eq:waiting_time} and \eqref{eq:pick-up} respectively.

To approximate the formulas \eqref{eq:waiting_time} and \eqref{eq:pick-up}, we apply the following \emph{change of variable}, 
\begin{equation}\label{eq:change_of_variable}
    -\frac{a_i}{\pi}\log
(1-\frac{\pi R_{i,\cdot}^2}{a_i})\mapsto R_{i,\cdot}^2.
\end{equation}
 where with a slight abuse of notation  we keep the same variable for the new radius. In fact, the left-hand side of \eqref{eq:change_of_variable} is approximately $R_{i,\cdot}^2$ when $R_{i,\cdot}<<\sqrt{a_i}$ as $\log(1-y)\approx-y$ when $y$ approaches 0.  Then,
 the drivers' waiting times and pick-up times can be approximated as:
\begin{equation}\label{eq:approximate}
    \begin{cases}
    w_i^d(x_{ii})&=\frac{-\log(1-\frac{x_{ii}}{b_i}e^{-\frac{b_i-x_{ii}}{\theta_i}\pi R_{i,c}^2})}{x_{ii}\pi R_{i,d}^2}\,,\\
    \tau_i(x_{ii})&=\frac{1}{v_i}\int_0^{R_{i,c}}e^{-\mu_i^c(x_{ii})\pi r^2}-e^{-\mu_i^c(x_{ii})\pi R_{i,c}^2}\dif r+\frac{b_i}{v_ix_{ii}}\int_0^{R_{i,d}}e^{-\mu_i^d(x_{ii})\pi r^2}-e^{-\mu_i^d(x_{ii})\pi R_{i,d}^2}\dif r\,.
\end{cases}
\end{equation}
With these formulas, we are able to analyze the equilibria of the game in the next section.

\subsection{Equilibria Analysis} 
\label{section:original_NN}
In this section, we show that the original NN dispatch algorithm, where the matching radius $R_{i,\cdot}$ is set to the maximal value $\sqrt{a_i/\pi}$, i.e., the approximated radius of region $i$, (or  $\infty$  if the approximated formulas \eqref{eq:approximate} are used), can  lead to multiple equilibria even when there is only one region. Furthermore, we examine the influence of the matching radius in the one-region case and discuss how these findings could be extended to multi-region scenarios.

We first  analyze the equilibrium when  there is only one region. To simplify notation, we omit the subscript $i$. As there is no driver repositioning, we use $x$ to denote the driver supply rate.
When the two matching radii are $\sqrt{a/\pi}$, the equation \eqref{eq:steady_state} becomes,
\[ x\cdot0^{\mu^c}+b\cdot0^{\mu^d}=b,\]
which seems to be meaningless. But if we define the indeterminate limit\footnote{Note that $\lim\limits_{x\rightarrow0}\lim\limits_{y\rightarrow0}x^y=1$  and $\lim\limits_{y\rightarrow0}\lim\limits_{x\rightarrow0}x^y=0$.} $0^0$ to be $1$, we obtain two possible solutions for $\mu^c,\mu^d$,
\[\mu^c=0,\mu^d>0\text{ or }\mu^c>0,\mu^d=0.\]
These solutions are validated through an intuitive argument: it is impossible for both the number of waiting customers and the number of waiting drivers to be positive when the matching radius is infinite, as they would be matched with each other. Consequently, there are two distinct types  of equilibria:
\begin{itemize}
    \item \emph{Oversupply Equilibrium}:  There is an oversupply of drivers, who must wait to be matched with customers. This results in a positive average waiting time for drivers.
    \item \emph{Undersupply Equilibrium}: There is an undersupply of drivers, leading to customers waiting for available drivers. This results in a positive average waiting time for customers.
\end{itemize}

We next analyze the two types of equilibria. As either $\mu^c=0,\mu^d>0\text{ or }\mu^c>0,\mu^d=0$, we assume the following form of square root law to estimate the average pick-up time,
\begin{equation}\label{eq:sq_root}
    \tau=\begin{cases}
\frac{1}{2v\sqrt{\mu^c}}&\text{ if }\mu^d=0,\\
 \frac{1}{2v\sqrt{\mu^d}}&\text{ if }\mu^c=0.
\end{cases}
\end{equation}
Similar forms of square root law are also used in  related literature \cite{Besbes2022,wang2022}. In Appendix \ref{appendix:pickup}, we demonstrate that  the formulas in \eqref{eq:pick-up} and \eqref{eq:approximate} both lead to \eqref{eq:sq_root} when  $R_{i,\cdot}\rightarrow\sqrt{a_i/\pi}$ and   $R_{i,\cdot}\rightarrow\infty$ respectively  through an asymptotical analysis. In the following discussion, we use  $x^o$ and $x^u$ to denote the equilibrium driver supply rates for oversupply type and
undersupply type respectively.

\textbf{Type 1: Oversupply Equilibrium.} In this type of equilibrium, there is an oversupply of drivers, resulting in a positive density of waiting drivers, i.e., $\mu^c = 0$ and $\mu^d > 0$. The driver supply rate in this equilibrium is always equal to $b$, i.e., $x^o = b$, as every customer is served in the oversupply equilibrium.
 The only unknown variable for this type of equilibrium is the density of waiting drivers, $\mu^d$. To find the unknown $\mu^d$, we
apply the pick-up time formula in  \eqref{eq:sq_root} and solve the  conservation law of mass equation as follows,
\begin{equation}\label{eq:ud}
   \mu^d + \frac{b}{2v\sqrt{\mu^d}} + bt = m\,,
\end{equation}
which has one  solution when
\(
m = \tfrac{1}{3} \left( \tfrac{b^2}{16v^2} \right)^{\frac{1}{3}} + bt,
\)
and two  solutions when
\(
m > \tfrac{1}{3} \left( \tfrac{b^2}{16v^2} \right)^{\frac{1}{3}} + bt.
\)

\textbf{Type 2: Undersupply Equilibrium.} In this type of equilibrium, there is an undersupply of drivers, resulting in a positive density of waiting customers, i.e., $\mu^c>0$ and $\mu^d=0$. 
According to Little's law, $w^d = \mu^d / x = 0$, as there are no waiting drivers. In this case, the unknown variable is the driver supply rate, $x$.  To find the unknown $x$, we
use the pick-up time formula in \eqref{eq:sq_root} and solve the  conservation law of mass equation \(\tfrac{x}{2v\sqrt{\mu^c}} + xt=m\). By substituting \eqref{eq:abandon}, the equation becomes
\begin{equation}\label{eq:xc}
   \frac{x}{2v}\sqrt{\frac{\theta}{b-x}} + xt = m,
\end{equation}
which has exactly one solution as long as $m > 0$. 
The equilibrium driver supply rate $x^u$ satisfies equation \eqref{eq:xc} and $x^u < b$, representing an inefficient equilibrium compared to the oversupply one. 

We summarize all the possible equilibria results   as follows. 

\begin{proposition}\label{prop:multiple_equilibria}
    Assume there is a single region and the platform adopts the NN dispatch algorithm with no restriction on the matching radius. There can be multiple equilibria depending on the size of driver density, $m$:
    \begin{itemize}
        \item[1)] When $0<m<\tfrac{1}{3}\left(\tfrac{b^2}{16v^2}\right)^{\frac{1}{3}}+bt$, there is only one  equilibrium which is an undersupply one with driver supply rate $x^c$ the unique solution to \eqref{eq:xc} and zero waiting driver density.
        \item[2)] When $m=\tfrac{1}{3}\left(\tfrac{b^2}{16v^2}\right)^{\frac{1}{3}}+bt$, there are two equilibria: one oversupply equilibrium with driver supply rate $x^o=b$ and waiting driver density $\mu^d$ the unique solution to \eqref{eq:ud}, and one undersupply equilibrium with driver supply rate $x^u$ the unique solution to \eqref{eq:xc}, and zero waiting driver density .
           \item[3)] When $m>\tfrac{1}{3}\left(\tfrac{b^2}{16v^2}\right)^{\frac{1}{3}}+bt$, there are three equilibria: two oversupply equilibrium with the same driver supply rate $x^o=b$ and two different waiting driver densities given by solutions to \eqref{eq:ud} respectively, and one undersupply equilibrium with driver supply rate $x^u$ the unique solution to \eqref{eq:xc} and zero waiting driver density.
    \end{itemize}
\end{proposition}

Proposition \ref{prop:multiple_equilibria} demonstrates that there can be multiple equilibria under the NN dispatch algorithm with no restriction on the matching radius even in a single region.  In case 2) and 3) of Proposition \ref{prop:multiple_equilibria},  oversupply and   undersupply equilibria coexist. The undersupply equilibria are  inefficient ones  (sometimes referred to as WGC equilibria): comparing \eqref{eq:ud} and \eqref{eq:xc}, the pickup time is longer in the undersupply equilibrium because a greater proportion of the mass is engaged in picking up customers (represented by the second term on the left-hand side of \eqref{eq:xc}).
In such inefficient equilibria, drivers are assigned to pick up customers located far away and the system experiences a driver shortage because: 1)
high average pick-up times force drivers to travel long distances to reach customers; 2)
wasted time en route further reduces driver availability, exacerbating the undersupply.
This creates a vicious cycle: longer travel times lead to fewer drivers being available to respond to new requests, worsening the imbalance between supply and demand. In the following, we show that theoretically the efficiency loss can be infinitely large.

\textbf{Infinity Price of Anarchy}. The undersupply equilibrium can be arbitrarily worse compared to the oversupply equilibrium, i.e.,  $x^u<<b$. Specifically, when the mass is $m=\tfrac{1}{3}\left(\tfrac{b^2}{16v^2}\right)^{\frac{1}{3}}+bt$, there are exactly two equilibria.  
The two equilibria are $x^o=b$ and $x^u$, where $\frac{x^u}{b}$  satisfies  
\[ \frac{x^u}{b}=\frac{\frac{1}{3}\Big(\frac{b^2}{16v^2}\Big)^{\frac{1}{3}}+bt}{\frac{b}{2v}\sqrt{\frac{\theta}{b-x^u}} + bt},\]
which approaches $0$ if $\theta\rightarrow\infty$. We summarize the results in the following corollary.
 \begin{corollary}\label{coro:poa}
     When $m=\tfrac{1}{3}\left(\tfrac{b^2}{16v^2}\right)^{\frac{1}{3}}+bt$ (case 2) of Proposition \ref{prop:multiple_equilibria}), the ratio between the two equilibrium driver supply  rates, $x^u$ and $x^o$, tends to 0 as the waiting customers' abandonment rate,  $\theta$, goes to infinity, i.e.,
     \[\lim_{\theta\rightarrow\infty}\frac{x^u}{x^o}=0.\]
 \end{corollary}

Therefore, the inefficient equilibrium ($x^u$) can be arbitrarily worse than the efficient equilibrium ($x^o$). 
Note that an infinity price of anarchy  occurs under the following conditions: 1)   the customers exhibit extreme impatience ($\t\rightarrow\infty$); 2) the region size is sufficiently large ($a\rightarrow\infty$ in order to apply the square root law \eqref{eq:sq_root} in this extreme case). This extreme scenario is primarily theoretical.  However, it highlights a critical insight: if the dispatch algorithm is not properly designed, the system’s efficiency loss can be catastrophic.

\begin{table}[!ht]  
\small
    \centering 
   \caption{Values of Parameters}
    \begin{tabular}{c|c|c|c|c|c|c|c}
    \thickhline 
    Parameters & $a$& $b$&$v$&$\frac{1}{\t}$&$t$&$m$&$R$\\ 
     \thickhline 
     Units & $\text{ km}^2$&  $\text{ pass}/(\text{min}\cdot\text{km}^2)$ & $\text{ km/min}$ & $\text{ min}$& $\text{ min}$ & $\text{ veh}/\text{km}^2$ &$\text{ km}$ \\ 
    \thickhline 
 Values &$100$&  $2$ & $0.5$ & $1$& $10$ & $(0,\infty)$ & $\{1,2,3,4\}$\\
     \thickhline 
    \end{tabular}\label{tab:par}
\end{table}

\begin{figure*}[htbp]
    \centering
    \includegraphics[width=1\linewidth]{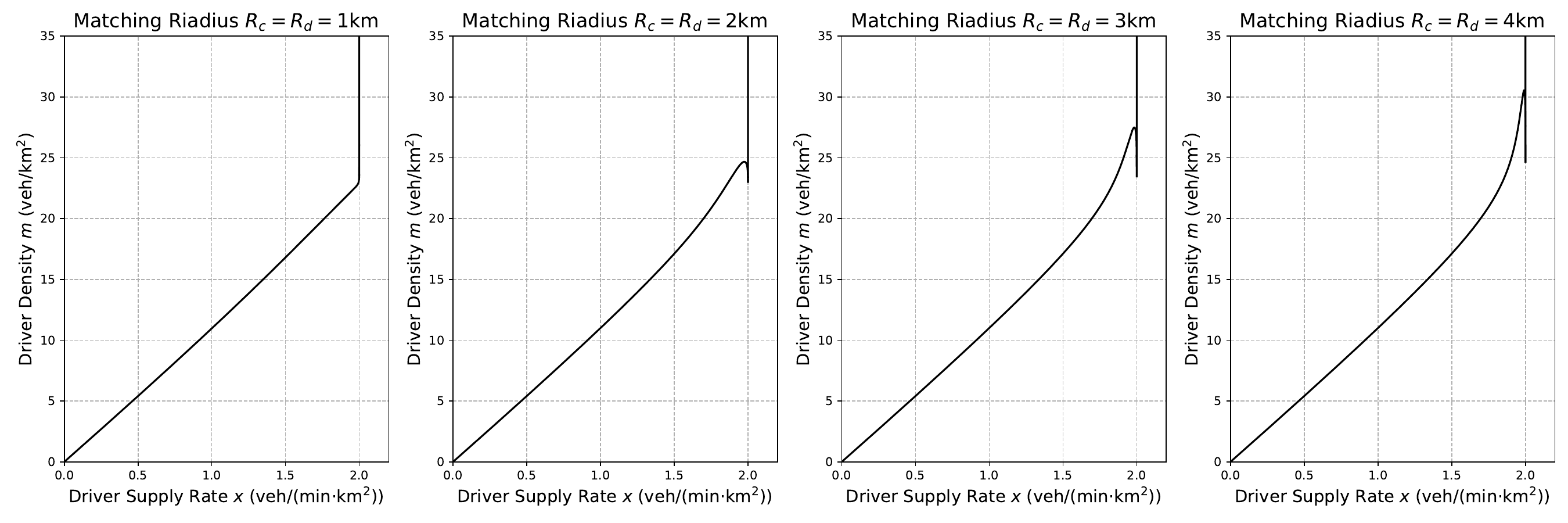}
    \caption{ The LHS of \eqref{eq:ne_oneregion} as a function of the driver supply rate $x$. We set the matching radius $R_c$ and $R_d$ to be the same.  In this figure, the matching radius values, from left to right, are 1 km, 2 km, 3 km, and 4 km. We can observe that for larger matching radii ($R=2,3,4$ km) and a driver density near 25, there can be up to three equilibria.  Additionally, the efficiency loss in the inefficient equilibrium is approximately $5\%-10\%$.} \label{fig:four_radii}
\end{figure*}

\textbf{Impact of Matching Radius.} In the following we investigate how different matching radii affect the outcome of equilibria and reveal the importance of the choice of matching radii. 
In the one region case,
we can compute the equilibria by solving the equation of conservation law of mass in \eqref{eq:conservation}, i.e.,
\begin{equation}\label{eq:ne_oneregion}
    x\big(w^d(x)+\tau(x)+t\big)=m,
\end{equation}
where the driver supply rate $x$ is the unknown and formulas of $w^d(x)$ and $\tau(x)$ are given by \eqref{eq:waiting_time} and \eqref{eq:pick-up}. Since $\tau(x)$   is non-monotonic, multiple equilibria can arise for specific parameter choices. Hence, the use of fixed matching radius with $R_c,R_d<\sqrt{a/\pi}$, does not completely eliminate multiple equilibria or prevent the inefficient equilibrium.

Due to the absence of a closed-form solution for \eqref{eq:ne_oneregion}, we employ a numerical approach to identify equilibria by plotting and analyzing the left-hand side (LHS) of \eqref{eq:ne_oneregion}. We use parameter values that are  outlined in Table \ref{tab:par}. 
In Figure \ref{fig:four_radii}, we depict the  LHS of \eqref{eq:ne_oneregion} as a function of the driver supply rate $x\in(0,b)$, with the matching radius set to $R_c=R_d=$ 1 km, 2 km, 3 km, and 4 km from left to right. The figure reveals that the LHS of \eqref{eq:ne_oneregion} generally  increases  monotonically with $x$, implying there is a single equilibrium (i.e., solution to \eqref{eq:ne_oneregion}) for most values of  driver density $m$. However, near  $b=2$ pass/(min$\cdot$km$^2$), the function becomes non-monotonic, leading to at most three equilibria. Although under normal parameter settings as in Table \ref{tab:par}, these multiple equilibria are closely spaced and pose minimal threat to system efficiency, extreme cases, as described in Corollary \ref{coro:poa}, as well as scenarios where the parameters approach those in Corollary \ref{coro:poa}, can lead to catastrophic efficiency losses.   Moreover, the range of driver densities over which there are multiple equilibria expands with the matching radius, suggesting a potential for greater efficiency loss when using larger matching radius. Thus,
the choice of matching radii significantly influences the equilibria and overall system efficiency, necessitating careful selection. These observations underscore the importance of optimal radii design. We will explore the optimal radii design explored in the next section.

\textbf{Extension to Multiple Region.} When multiple regions are considered, the average pick-up time in each region becomes a non-monotonic function of the local driver supply rate, potentially resulting in a greater number of equilibria. For the single-region case, there may be up to three equilibria, but as the number of regions increases, the number of equilibria grows combinatorially. Due to the complexity of the detailed analysis, we omit it here, though this phenomenon can be verified using a simplified model, such as a spoke-and-hub network.

An important observation for systems with multiple regions is as follows: if the sum of the driver's average waiting time, $w_i^d$, and the pick-up time, $\tau_i$ is an increasing function of the driver supply rate $x_{ii}$ for any region $i\in S$, then by Corollary \ref{coro:uniqueness}, there is a unique equilibrium. However, under a general choice of matching radii $(R_{\cdot,c}, R_{\cdot,d})$, $w_i^d+\tau_i$ is not necessarily an increasing function of  $x_{ii}$. In the next section, we address this challenge and obtain monotonicity by dynamically optimizing the matching radii.

\section{Dynamically Optimize Matching Radii}\label{section:optimaldispatch}
In this section, we demonstrate how to eliminate multiple equilibria by implementing an optimal adjustment of matching radii.

The multiple equilibria are not resolved by introducing fixed matching radii. Instead of treating matching radii as fixed parameters of the dispatch algorithm, we propose to dynamically optimize them.  Specifically, we treat the waiting time $w_i^d$ and pick-up time $\tau_i$ in \eqref{eq:approximate} as functions of both the driver supply rate $x_{ii}$ and the matching radii $R_{i,c}$ and $R_{i,d}$, and use the optimal radii  $R_{i,c}$ and $R_{i,d}$  that minimize the sum of waiting and pick-up times:
\begin{equation}\label{eq:dynamic_radii}
    (R_{i,c}(x_{ii}),R_{i,c}(x_{ii}))=\arg\min_{(R_{i,c},R_{i,d})\in\mathbb{R}^2_{+}}w_i^d(x_{ii},R_{i,c},R_{i,d})+\tau_i(x_{ii},R_{i,c},R_{i,d}),\quad\forall i\in S.
\end{equation}
Here, $w_i^d$ and $\tau_i$ are given by approximation formulas in \eqref{eq:approximate} since the original formulas  \eqref{eq:waiting_time} and \eqref{eq:pick-up} can be approximated by \eqref{eq:approximate} when $R_c,R_d$ are small.
We show that the above optimization problem in \eqref{eq:dynamic_radii} can be solved analytically and the dynamic matching radii function $(R_{i,c}(x_{ii}),R_{i,d}(x_{ii}))$ is well-defined as the unique minimizer of $w_i^d+\tau_i$ for each $x_{ii}$, in the following proposition.

\begin{proposition}[\textbf{Optimal Radii Design}]\label{prop:optimal_rdaius}
For any $i\in S$, the optimization problem 
\begin{equation}\label{eq:opt}
\min_{(R_{i,c},R_{i,d})\in\mathbb{R}^2_{+}}w_i^d(x_{ii},R_{i,c},R_{i,d})+\tau_i(x_{ii},R_{i,c},R_{i,d}),
\end{equation}
with  $w_i^d$ and $\tau_i$ given in \eqref{eq:approximate} has a unique optimal solution $(R_{i,c}^*,R_{i,d}^*)$ with $R_{i,c}^*=R_{i,d}^*=R_i^*(x_{ii})$ where $R_i^*(x_{ii})$ is the unique solution to the following equation regarding $R_i$,
\begin{align}\label{eq:opt_sol}
     &\left(\frac{b_i}{4v_i}\right)^{\frac{2}{3}}\left(\erf\sqrt{-\log(1-\frac{x_{ii}}{b_i}e^{-\frac{b_i-x_{ii}}{\theta_i}\pi R_{i}^2})}-\frac{2}{\sqrt\pi}\sqrt{-\log(1-\frac{x_{ii}}{b_i}e^{-\frac{b_i-x_{ii}}{\theta_i}\pi R_{i}^2})}(1-\frac{x_{ii}}{b_i}e^{-\frac{b_i-x_{ii}}{\theta_i}\pi R_{i}^2})\right)^{\frac{2}{3}}\\
        &\quad=\frac{-\log(1-\frac{x_{ii}}{b_i}e^{-\frac{b_i-x_{ii}}{\theta_i}\pi R_{i}^2})}{\pi R_{i}^2} \notag
\end{align}
Here, $\erf$ is the error function defined by $\erf(y)=\tfrac{2}{\sqrt{\pi}}\int_0^ye^{-s^2}\dif s$.
\end{proposition}

The proof can be found in Appendix \ref{appendix:proof_prop:optimal_rdaius}. Proposition \ref{prop:optimal_rdaius} show that for any driver supply rate $x_{ii}>0$ in any region $i\in S$, the optimal choice of matching radii is unique and can be found by solving \eqref{eq:opt_sol}.
In the following theorem, we demonstrate that the sojourn time $T_{ii}(x_{ii})$ as defined in Assumption \ref{assum:function_of_x} is indeed a strictly increasing function of $x_{ii}$ for all $i\in\mathcal{S}$ if the proposed dispatch algorithm with optimal adjustment of matching radii in  \eqref{eq:dynamic_radii} is adopted.
\begin{theorem}[\textbf{Monotone Sojourn Time}]\label{thm:monotone}
    Given the optimal adjustment of matching radii  $(R_{i,c}(x_{ii}),R_{i,d}(x_{ii}))$ defined by \eqref{eq:dynamic_radii}, i.e., $R_{i,c}(x_{ii})=R_{i,d}(x_{ii})=R_i^*(x_{ii})$, where $R_i^*(x_{ii})$ is the unique solution to \eqref{eq:opt_sol}, the  function 
    \[f_i(x_{ii})=w_i^d(x_{ii},R_{i,c}(x_{ii}),R_{i,d}(x_{ii}))+\tau_i(x_{ii},R_{i,c}(x_{ii}),R_{i,d}(x_{ii}))\]
   is a strictly increasing in $x_{ii}$ when $x_{ii}\in(0,b_i)$ for all  $i\in\mathcal{S}$. Consequently,   there exists a unique stationary equilibrium for game $\mathcal{G}$.
\end{theorem}

The proof can be found in Appendix \ref{appendix:proof_thm:monotone}.  Theorem \ref{thm:monotone} establishes that the optimal adjustment of matching radii, as defined in \eqref{eq:dynamic_radii}, ensures a unique equilibrium, thereby eliminating the inefficient equilibrium with long pick-up times. Importantly, this theorem also demonstrates that our optimal adjustment guarantees the optimization problem \eqref{eq:opt_obj} is strictly convex--a property whose significance has not been sufficiently emphasized so far.
When the optimization problem \eqref{eq:opt_obj} is nonconvex, multiple equilibria may arise, and computing these equilibria (or KKT points) becomes computationally challenging. In contrast, the strict convexity of \eqref{eq:opt_obj} enables the use of efficient gradient-based numerical methods to compute the unique equilibrium, enhancing the stability and predictability of the ride-hailing system.
Moreover, by minimizing the driver's total waiting time (the sum of waiting and pick-up times) in each region, our proposed algorithm also enhances the efficiency of the ride-hailing system. The improvement is further confirmed through the numerical experiments and simulations conducted for a single region.

\begin{figure}[h]
    \begin{minipage}{0.32\textwidth}
        \includegraphics[width=\textwidth]{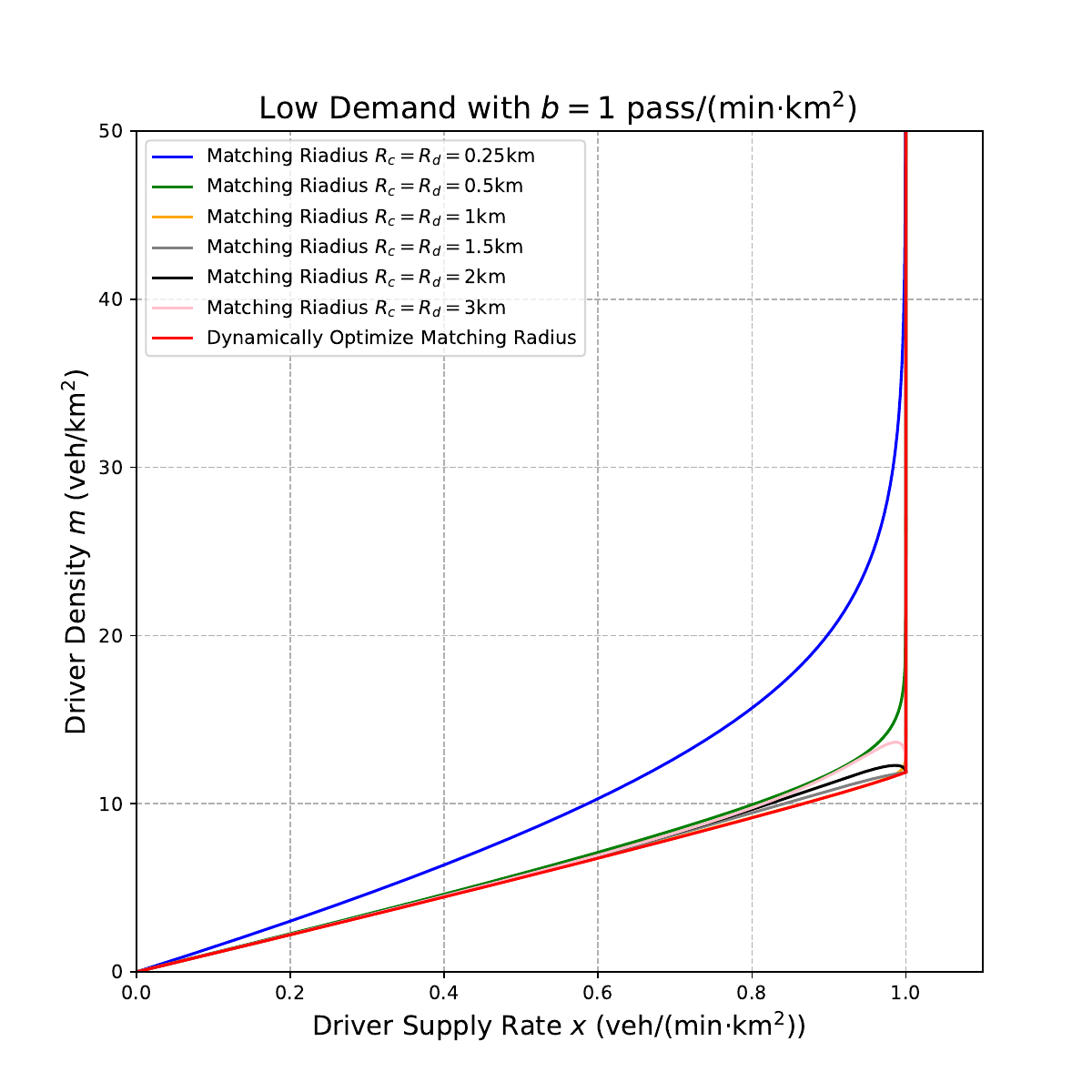}
    \end{minipage}
    \begin{minipage}{0.32\textwidth}
        \includegraphics[width=\textwidth]{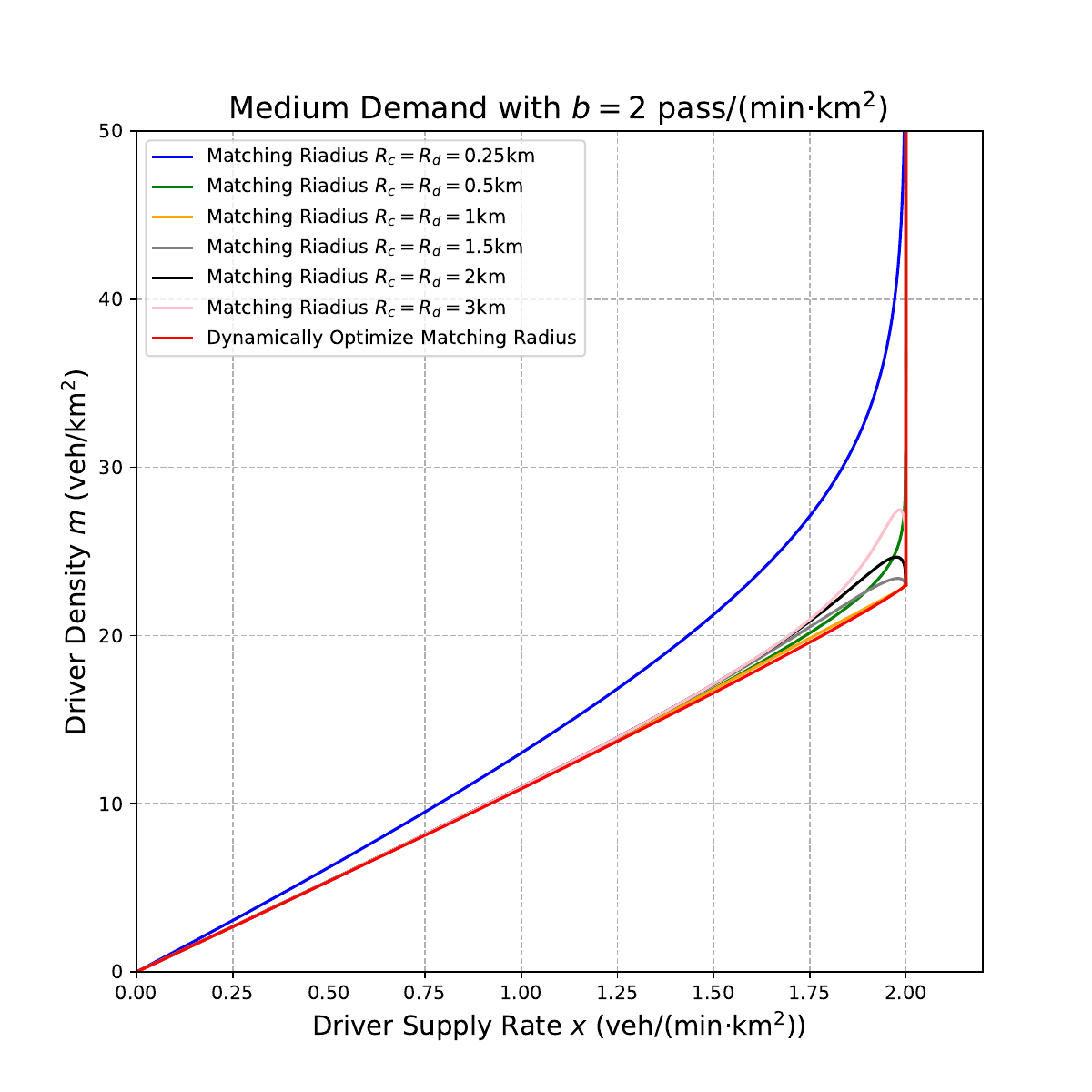}
    \end{minipage}
    \begin{minipage}{0.32\textwidth}
        \includegraphics[width=\textwidth]{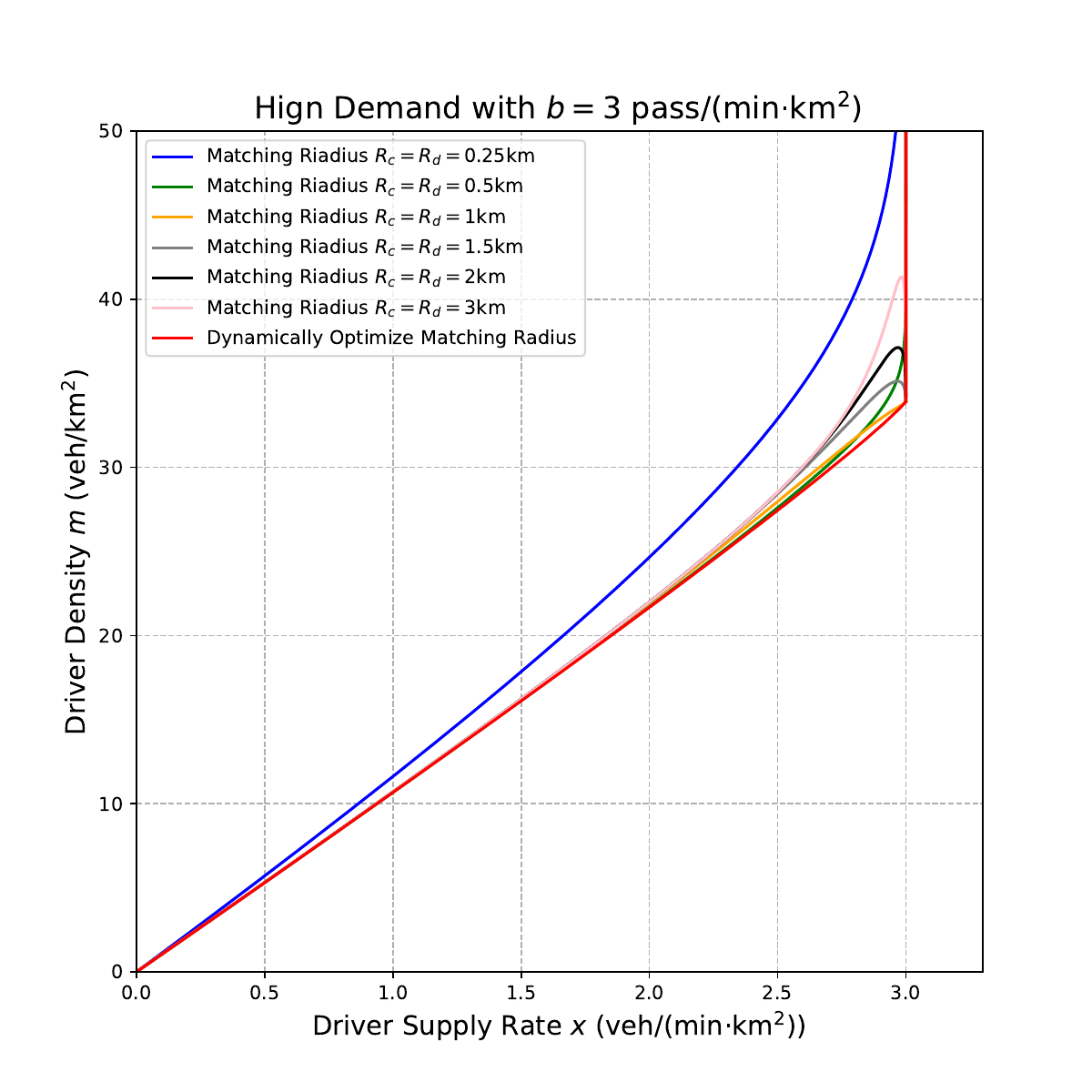}
    \end{minipage}
    \caption{Figure showing driver density as a function of driver supply rate (RHS of \eqref{eq:ne_oneregion}), with various customer arrival rates and matching radii, compared to the proposed dynamic optimal matching radius \eqref{eq:dynamic_radii}. Panels from left to right represent low, medium, and high demand scenarios with customer arrival rates of  $1$, $2$, $3$ pass/(min$\cdot$km$^2$), respectively. The proposed dynamic optimal matching radius (highlighted by the red color) outperforms all fixed radii, and under this dynamic scheme, driver density strictly increases with the driver supply rate, ensuring a unique equilibrium.}
    \label{fig:dynamic}
\end{figure}

\textbf{Numerical Test}  We conduct a numerical study in a single-region scenario as described in Section \ref{section:original_NN}. Fig. \ref{fig:dynamic} illustrates driver density as a function of driver supply rate (RHS of \eqref{eq:ne_oneregion}), across various customer arrival rates and matching radii, and compares these with the proposed  optimal adjustment of matching radius \eqref{eq:dynamic_radii}. The sub-figures, from left to right, depict low, medium, and high demand scenarios with customer arrival rates of  $1$, $2$, $3$ pass/(min$\cdot$km$^2$), respectively. The values of other parameters are the same as in Table \ref{tab:par}.  We observe that no dispatch algorithm with fixed radius remains optimal across all driver supply rates. The proposed algorithm with optimal adjustment of matching radius, highlighted in red, surpasses the performance of all algorithms with fixed radii. Specifically, for each fixed driver density, the equilibrium driver supply rate under the optimal adjustment of matching radius scheme is always the highest.
Furthermore, under this optimal adjustment of matching radius scheme, driver density exhibits a strict increase with the driver supply rate, thereby ensuring a unique equilibrium.

\section{Simulation}

\begin{table}[h!]\footnotesize
  \begin{tabular}{m{1.5cm}<{\centering}m{1.25cm}<{\centering}m{1.25cm}<{\centering}m{1.25cm}<{\centering}m{1.25cm}<{\centering}m{1.25cm}<{\centering}m{1.25cm}<{\centering}m{1.25cm}<{\centering}m{1.25cm}<{\centering}}
      \toprule[1pt]
   Matching radius $R$ (km) &   Total number of arrival customers & Mean waiting time for customers (min) & Mean pick-up time (min) & Mean waiting time for drivers (min) & Mean total waiting time for customers (min)  & Mean total waiting time for drivers (min) & Completion rate  \\
      \midrule[0.5pt]
 0.5&14240 &4.20& \cc{0.78}& 13.85& 4.98& 14.63& 0.591\\
 1.0&14192 &2.73& 1.09& 7.01& 3.82& 8.10& 0.728\\
 1.5&14309 &2.10& 1.43& 4.12& 3.53& 5.55& 0.794\\
 2.0&14501 &1.72& 1.72& 2.68& \cc{3.44}& 4.40& 0.826 \\
 2.5&14329 &1.60& 2.04& 1.64& 3.64& 3.68& 0.842\\
 3.0&14229 &1.33& 2.33& 1.33& 3.66& 3.66& 0.868\\
3.5& 14454 &1.49& 2.56& 0.81& 4.05& \cc{3.37}& 0.853\\
 4.0&14275 &1.33& 2.84& 0.67& 4.17& 3.51& 0.864\\
 4.5&14557 &1.58& 2.92& 0.53& 4.50& 3.45& 0.842\\
 5.0&14292 &1.45& 3.16& \cc{0.43}& 4.61& 3.59& 0.857\\
Dynamic Adjustment of Radius& 14355&\cc{1.27}&2.40&1.21&3.67&3.61&\cc{0.870}\\
    \bottomrule[1pt]
  \end{tabular}
  \caption{Comparison of Performance Metrics for Fixed Matching Radii and Dynamic Adjustment of Radius}
  \label{table:comprehensive}
\end{table}

We have theoretically proved and numerically shown that the proposed optimal adjustment of matching radius  ensures a unique equilibrium and improve efficiency. To further verify the effectiveness of the proposed method, we
simulate the stochastic online matching process in a single region. We  apply the nearest-neighbor dispatch with the optimal adjustment of matching radius where the optimal matching radius is computed for each value of driver supply rate  using the equation \eqref{eq:opt_sol}. The experiments are performed using Python
3.11 on a 16-core CPU.

The simulator is established on a space with 10 km $\times$ 10 km area. The basic setting are as follows.
The simulation duration is 24 hour, i.e., 1440 min. The arrival time of customer follows a exponential distribution with rate 10 pass/min. Their locations are randomly generated. Customers are impatient and abandon the orders if the waiting time (until they are matched with the drivers) exceeds the maximum time they are willing to wait, which is exponentially distributed with mean 10 min.  There are in total 200 drivers and the vehicle speed is 0.4 km/min, i.e., 24 km/h. The travel time from customer's origin to customer's destination is exponentially distributed with mean 20 min. 
 After dropping a passenger, a vehicle becomes available for matching immediately and the  its new location is randomly generated.

The optimal adjustment of matching radius in the simulation works as follows: we solve \eqref{eq:opt_sol}  to determine the optimal
matching radius 
each time when a customer sends a request or when a driver becomes available where the driver supply rate is estimated by the average driver supply rate over the last hour. 
 We compare the performance of  optimal adjustment of matching radius and benchmark scenarios with fixed matching radii $R_c=R_d=R$  where the values of  $R$ is taken from 0.5 km to 5 km with a step size of 0.5 km. Evaluation metrics include the mean
pick-up time, mean waiting time, mean total waiting time (sum of pick-up time and waiting time) of passengers and drivers, and the
completion rate (the proportion of passenger requests served successfully). The results in Table \ref{table:comprehensive} show that the dynamic
optimization strategy achieves a higher completion rate (0.870) than all the benchmark scenarios. The strategy also
achieves a reasonably low mean total waiting time (pick up time plus waiting time) for both customers and drivers. Overall, we can claim that  optimal adjustment of matching radius achieves near optimal system performance among all the simulated cases.

\section{Conclusion}

In this study, we have made significant strides in understanding and enhancing the operational dynamics of ride-hailing platforms through the lens of mean-field game theory. Our primary contribution lies in the introduction of a novel mean field game model, which not only systematically analyzes the intricate dynamics of these platforms but also establishes a profound connection between the equilibrium points of the game and the KKT points of an associated optimization problem. This theoretical framework offers a fresh perspective on platform operations, enabling a deeper understanding of the underlying mechanisms that govern driver and customer interactions.
Furthermore, our analysis of  ride-matching algorithms in the presence of customer abandonment behavior has unveiled critical insights. We have identified that these algorithms can inadvertently lead to non-monotone sojourn times, resulting in the emergence of multiple equilibria and the problematic inefficient ones. The inefficient equilibria, characterized by drivers engaging in inefficient pursuits of distant requests, pose a significant challenge to the overall efficiency of ride-hailing systems.
To address these inefficiencies, we have proposed an innovative two-matching-radius nearest-neighbor dispatch algorithm. Our rigorous theoretical proofs and empirical validations demonstrate that the dynamic and optimal adjustment of matching radii can effectively ensure that sojourn times increase monotonically with the driver supply rate. This breakthrough not only provides a robust solution to the multiple and inefficient equilibria problem but also enhances the overall efficiency and reliability of ride-hailing platforms.

We conclude this paper by envisioning a roadmap for future research that extends the boundaries of our current work. One promising avenue lies in the exploration of pricing mechanisms within ride-hailing platforms. Inspired by \cite{castillo2023}'s findings, which indicate that surge pricing outperforms matching adjustments in mitigating the inefficient equilibria, we aim to conduct a comprehensive and theoretical comparison using our mean field game framework. The findings in current work suggest that an optimal matching strategy can effectively eliminate inefficient equilibria without reliance on any  pricing designs.  We expect that incorporating pricing design into our model will unveil new dimensions and insights, further enriching our understanding of ride-hailing dynamics.
Another extension of our model is its application to delivery problems, where couriers, analogous to drivers in our setting, have the capability to deliver multiple parcels simultaneously. The analysis of spatial queueing networks in this context is more complex and challenging. We are particularly interested in exploring whether our approach of expressing sojourn times as functions of throughput can be adapted and applied to this more intricate setting. Such an extension could provide valuable insights into the optimization of delivery systems, contributing to enhanced efficiency and service quality.


\bibliographystyle{ACM-Reference-Format}
\bibliography{biblio-1}

\newpage
\appendix

\section{Derivation of Formulas for Waiting Times and Pick-up Times
}\label{appendix:proof_formulas}

In this appendix we establish formulas for the average waiting time and pick-up time as claimed in Assumption \ref{assum:function_of_x} based on the customer abandonment behavior defined in Section \ref{section:Abandonment} and the platform’s dispatching algorithm defined in Section \ref{section:Matching}.

 Let  each region $i$'s area be denoted by $a_i$.  
We introduce two matching radii in region $i$, $R_{i,c},R_{i,d}\in(0,\sqrt{a_i/\pi})$, where  $\sqrt{a_i/\pi}$ is the approximate radius of region $i$.   
 We denote the average vehicle speed by $v_i$. Let $\mu_i^c$ and $\mu_i^d$ denote the average densities of waiting customers and waiting drivers respectively across the region $i$. Recall that $x_{ii}$ denotes the rate of drivers that  become available and choose to serve region $i$ per unit area, and $b_i$, $\theta_i$ the arrival rate, abandonment rate of customers per unit area, respectively.
 
We assume the locations of waiting customers and waiting drivers are drawn from independent uniform distributions across region \(i\) at any moment. This assumption is justified by the randomness of customers' origins and destinations, as well as the mobility of the vehicles. Given the matching radius $R_{i,c}$, when a driver becomes available in region \(i\) and decides to provide service in this region, the probability that the driver is within distance \(R_{i,c}\) of a randomly chosen waiting customer in region \(i\) is \(\tfrac{\pi R_{i,c}^2}{a_i}\). Assume there are $N_i^c$ waiting customers in region $i$ at this moment, then the driver is not matched with any of these waiting customers with probability
\[(1-\frac{\pi R_{i,c}^2}{a_i})^{N_i^c}.\]
Similarly, given the matching radius $R_{i,d}$, if there are 
$N_i^d$ waiting drivers in region $i$, a new arriving customer is not matched with any of these waiting customers with probability
\[(1-\frac{\pi R_{i,d}^2}{a_i})^{N_i^d}.\]

Now we can find the formula for driver's average waiting time. A key observation is that for drivers, the rate at which they enter the waiting status equals the rate at which they exit it. This exit rate is equivalent to the rate at which arriving customers are matched with drivers.
We apply  mean-field approximation to replace stochastic variables, such as $N_i^c$ and $N_i^d$,  with their expected values: $\mu_i^ca_i$ and $\mu_i^da_i$ respectively. Then, in steady state, the relation between $\mu_i^c$  and $\mu_i^d$ is
\begin{equation}
    x_{ii}(1-\frac{\pi R_{i,c}^2}{a_i})^{\mu_i^ca_i}+b_i(1-\frac{\pi R_{i,d}^2}{a_i})^{\mu_i^da_i}=b_i\,.\label{eq:steady_state_ap}
\end{equation}
Note that $\mu_i^c$ and $x_{ii}$ satisfy \eqref{eq:abandon}, $\mu_i^c$ can be expressed as a function of $x_{ii}$,
\begin{equation}
    \label{eq:u_c_ap}
    \mu_i^c(x_{ii})=\frac{1}{\theta}(b_i-x_{ii}).
\end{equation}
As a result, by \eqref{eq:steady_state}, the unknown variable $\mu_i^d$ can also be expressed as a function of $x_{ii}$, 
\begin{equation}
    \label{eq:u_d}
   \mu_i^d(x_{ii})=\frac{\log(1-\frac{x_{ii}}{b_i}(1-\frac{\pi R_{i,c}^2}{a_i})^{\frac{b_i-x_{ii}}{\theta_i}a_i})}{a_i\log(1-\frac{\pi R_{i,d}^2}{a_i})}\,.
\end{equation}
Let $w_i^d$ denote the average waiting time of the drivers. By Little's law, $\mu_i^d=x_{ii}w_i^d$. Hence, $w_i^d$ can be expressed as a function of $x_{ii}$,
\begin{equation}\label{eq:waiting_time_ap}
    w_i^d(x_{ii})=\frac{\log(1-\frac{x_{ii}}{b_i}(1-\frac{\pi R_{i,c}^2}{a_i})^{\frac{b_i-x_{ii}}{\theta_i}a_i})}{x_{ii}a_i\log(1-\frac{\pi R_{i,d}^2}{a_i})}\,.
\end{equation}
One can check that $w_i^d(x_{ii})\rightarrow\infty$ as $x_{ii}\rightarrow b_i$ as we claimed in Assumption \ref{assum:function_of_x}.

Next, we derive the formula for expected pick-up time, which amounts to finding the expected pick-up distance. We compute this in two distinct cases:
\begin{itemize}
    \item \emph{Matched with no waiting}.
   In this case, a driver arrives and is successfully matched with a waiting customer. The probability of this occurring is 
    \[1-(1-\frac{\pi R_{i,c}^2}{a_i})^{\mu_i^ca_i}.\]
   If matched, the probability that the pick-up distance is less than $r$ is the probability that at least one waiting customer is within distance $r\in(0,R_{i,c}]$ of the driver, which is 
    \[1-(1-\frac{\pi r^2}{a_i})^{\mu_i^ca_i}\,.\]
    Consequently, the conditional expectation of the pick-up distance is 
    \[\frac{1}{1-(1-\frac{\pi R_{i,c}^2}{a_i})^{\mu_i^ca_i}}\int_0^{R_{i,c}}r\dif\Big(1-(1-\frac{\pi r^2}{a_i})^{\mu_i^ca_i}\Big)\]
    \item   \emph{Matched with waiting}. In this case, a driver arrives and must wait until they are matched with a newly arriving customer. The probability of this occurring is,
    \[(1-\frac{\pi R_{i,c}^2}{a_i})^{\mu_i^ca_i}.\]
    The expected pick-up distance can be computed from the perspective of the arriving customer. The probability that a customer is successfully matched with a waiting driver is,
    \[1-(1-\frac{\pi R_{i,d}^2}{a_i})^{\mu_i^da_i}.\]
   The probability that the pick-up distance is less than $r$ is the probability that at least one waiting driver is within distance $r\in(0,R_{i,d}]$ of the customer, which is 
    \[1-(1-\frac{\pi r^2}{a_i})^{\mu_i^da_i}\,.\]
Consequently, the conditional expectation of the pick-up distance is
    \[\frac{1}{1-(1-\frac{\pi R_{i,d}^2}{a_i})^{\mu_i^da_i}}\int_0^{R_{i,d}}r\dif\Big(1-(1-\frac{\pi r^2}{a_i})^{\mu_i^da_i}\Big)\]
    
\end{itemize}
In summary, the expected pick-up time is a function of $x_{ii}$ as follows,
\begin{equation}
    \label{eq:pick-up_ap}
    \tau_i(x_{ii})=\frac{1}{v_i}\int_0^{R_{i,c}}(1-\frac{\pi r^2}{a_i})^{\mu_i^ca_i}-(1-\frac{\pi R_{i,c}^2}{a_i})^{\mu_i^ca_i}\dif r+\frac{b_i}{v_ix_{ii}}\int_0^{R_{i,d}}(1-\frac{\pi r^2}{a_i})^{\mu_i^da_i}-(1-\frac{\pi R_{i,d}^2}{a_i})^{\mu_i^da_i}\dif r,
\end{equation}
where $\mu_i^c$ and $\mu_i^d$ are functions of $x_{ii}$ as in  \eqref{eq:u_c} and \eqref{eq:u_d} respectively.

\textbf{Approximate Formulas}. As the formulas for waiting time and pick-up time are too complex, we derive approximate formulas below.
We apply a \emph{change of variable} trick, 
\begin{equation}\label{eq:change_of_variable_ap}
    -\frac{a_i}{\pi}\log
(1-\frac{\pi R_{i,\cdot}^2}{a_i})\mapsto R_{i,\cdot}^2,
\end{equation}
where we keep the same notation for the new radius with a slight abuse of notation. Then,
\eqref{eq:steady_state} can be rewritten as,
\[   x_{ii}e^{-\mu_i^c\pi R_{i,c}^2}+b_ie^{-\mu_i^d\pi R_{i,d}^2}=b_i.\]
Furthermore, the drivers' waiting times and pick-up times can be approximated as:
\begin{equation}\label{eq:approximate_ap}
    \begin{cases}
    w_i^d(x_{ii})&=\frac{-\log(1-\frac{x_{ii}}{b_i}e^{-\frac{b_i-x_{ii}}{\theta_i}\pi R_{i,c}^2})}{x_{ii}\pi R_{i,d}^2}\,,\\
    \tau_i(x_{ii})&=\frac{1}{v_i}\int_0^{R_{i,c}}e^{-\mu_i^c(x_{ii})\pi r^2}-e^{-\mu_i^c(x_{ii})\pi R_{i,c}^2}\dif r+\frac{b_i}{v_ix_{ii}}\int_0^{R_{i,d}}e^{-\mu_i^d(x_{ii})\pi r^2}-e^{-\mu_i^d(x_{ii})\pi R_{i,d}^2}\dif r\,.
\end{cases}
\end{equation}
Note that  the left-hand side of \eqref{eq:change_of_variable} is approximately $R_{i,\cdot}^2$ when $R_{i,\cdot}<<\sqrt{a_i}$ as $\log(1-y)\approx-y$ when $y$ approaches 0 and goes to $\infty$ as $\pi R_{i,\cdot}^2\rightarrow a_i$. Hence,
when calculating the pick-up time we do not multiply the change of variable factor when computing the integral.  Numerical results show that this is still a good estimation of the original pick-up time even when the original radius $R_i$ is large.
There is an alternative point of view which yields the same approximate formula in \eqref{eq:approximate}: assume each region $i$ is an unbounded plane with matching radii $R_{i,c},R_{i,d}\in(0,\infty)$.  In this unbounded plane, the locations of waiting customers and drivers follow \emph{spatial Poisson processes}, where $\mu_i^c$ and $\mu_i^d$ are the waiting customer and waiting driver intensity parameters, respectively. In fact, the change of variable in equation \eqref{eq:change_of_variable} defines a map from the original finite region to an infinite region. Although this transformation distorts the Euclidean metric in the original region, the distortion is negligible when the matching radius is sufficiently small. For large matching radii, while the distortion cannot be entirely ignored, the introduced approximation error remains acceptable according to numerical results.

\section{Approximate Pick-up Time by Square Root Law} \label{appendix:pickup}
 When $R_{c}=R_{d}=\sqrt{a/\pi}$ (or $\infty$), our formulae for $\tau_i$ in \eqref{eq:pick-up} and \eqref{eq:approximate} can still be used to estimate the pick-up time. In fact, both formulas give the same asymptotic result. 
We assume $\mu^c=0,\mu^d>0$. In this case, $x=b$ as $\mu^c=0$. 
Applying \eqref{eq:pick-up}, the  average pick-up time is estimated by
\[\tau(x)=\frac{1}{v}\int_0^{\sqrt{\frac{a}{\pi}}}(1-\frac{\pi r^2}{a})^{\mu^dA}\dif r\overset{\frac{\pi r^2}{a}\mapsto s^2}{=}\frac{1}{v}\sqrt{\frac{a}{\pi}}\int_0^{1}(1-s^2)^{\mu^da}\dif s = \left(\frac{1}{\sqrt{\pi}}\int_0^{1}(1-s^2)^{\mu^da}\dif s\right)\frac{\sqrt{a}}{v}.\]
Using the change of variable trick (i.e., \eqref{eq:approximate}), the  average pick-up time is 
\begin{equation}
    \label{eq:sq_root_ap}
    \tau(x)=\frac{1}{v}\int_0^{\infty}e^{-\mu^d\pi r^2}\dif r=\frac{1}{2v\sqrt{\mu^d}}=\frac{1}{2\sqrt{\mu^da}}\frac{\sqrt{a}}{v},
\end{equation}
which is exactly the famous square root law in literature (e.g., see \cite{Besbes2022}). Recall that $\sqrt{\frac{a}{\pi}}$ estimates the radius  of the whole region and $\mu^d a$ is the mean of the number of waiting drivers ($N^d$) in the region.
We can numerically check that for $n\in \mathbb{N}^+$,
\[\left(\frac{1}{\sqrt{\pi}}\int_0^{1}(1-s^2)^{n}\dif s\right)-\frac{1}{2\sqrt{n}}\]
is small even for small $n$ and goes to 0 as $n\rightarrow\infty$. 
In fact, for $y\in(0,\infty)$
\[\int_0^{1}(1-s^2)^{y}\dif s\overset{s^2\mapsto s}{=}\frac{1}{2}\int_0^{1}s^{-\frac{1}{2}}(1-s)^{y}\dif s=\frac{1}{2}B(\frac{1}{2},y+1)=\frac{1}{2}\frac{\Gamma(\frac{1}{2})\Gamma(y+1)}{\Gamma(y+\frac{3}{2})}\overset{y\rightarrow\infty}{\rightarrow}\frac{\sqrt{\pi}}{2\sqrt{y}},\]
where $B(y_1,y_2)=\int_0^{1}s^{y_1-\frac{1}{2}}(1-s)^{y_2-1}\dif t$ is the Euler's beta function and $\Gamma(y)=\int_0^{\infty}s^{y-1}e^{-s}\dif s$ is the gamma function.
Similarly, when $\mu^c>0,\mu^d=0$, $\tau(x)\rightarrow\tfrac{1}{2v\sqrt{\mu^c}}$ if $\mu^ca=\frac{b-x}{\t}a\rightarrow\infty$. Therefore, 
we will use the square root law \eqref{eq:sq_root} to estimate the pick-up time.

\section{Proof of Proposition \ref{prop:optimal_rdaius}} \label{appendix:proof_prop:optimal_rdaius}

\begin{proof}
    The optimization problem \eqref{eq:opt} is equivalent to the following problem,
    \begin{align}
        \min&\quad \frac{u_i}{x_{ii}}+\frac{1}{2v_i}\sqrt{\frac{\theta_i}{b_i-x_i}}\left(\erf\alpha_i-\frac{2}{\sqrt\pi}\a_i e^{-\a_i^2}\right)+\frac{b_i}{2v_ix_{ii}\sqrt {u_i}}\left(\erf\b_i-\frac{2}{\sqrt\pi}\b_i e^{-\b_i^2}\right)\label{eq:opt_sim}\\
        \text{s.t.}&\quad x_{ii}e^{-\a_i^2}+b_ie^{-\b_i^2}=b_i\notag\\
        \text{over}&\quad u_i,\a_i,\b_i\in\mathbb{R}_+\notag
    \end{align}
    which is obtained from \eqref{eq:opt}  by  changes of variables as follows,
\begin{align*}
    \left\{
    \begin{aligned}
              \frac{-\log(1-\frac{x_{ii}}{b_i}e^{-\frac{b_i-x_{ii}}{\theta_i}\pi R_{i,c}^2})}{\pi R_{i,d}^2}& \mapsto u_i\\
    \frac{b_i-x_{ii}}{\theta_i}\pi R_{i,c}^2 &\mapsto \a_i^2\\
          -\log(1-\frac{x_{ii}}{b_i}e^{-\frac{b_i-x_{ii}}{\theta_i}\pi R_{i,c}^2}) &\mapsto \b_i^2
    \end{aligned}
    \right.
\end{align*}
Note that the objective function of  \eqref{eq:opt_sim} is smooth for $u_i,\a_i,\b_i>0$. The first order condition of \eqref{eq:opt_sim} gives exactly $R_{i,c}=R_{i,d}$ and equation \eqref{eq:opt_sol}. One can check the following facts. 1) As $u_i,\a_i,\b_i$ $\rightarrow0$ or $\rightarrow\infty$, the asymptotical values of the objective function of \eqref{eq:opt_sim} are suboptimal. Consequently, equation \eqref{eq:opt_sol} is also a sufficient condition.  2) The equation \eqref{eq:opt_sol} has a unique solution in $\mathbb{R}_+$. Hence, the optimal solution to \eqref{eq:opt_sim} can be obtained by solving the equation \eqref{eq:opt_sol}. The proposition is proven.
\end{proof}

\section{Proof of Theorem \ref{thm:monotone}} \label{appendix:proof_thm:monotone}

\begin{proof}
    We prove the theorem using the equivalent optimization problem \eqref{eq:opt_sim}. As the objective function of  \eqref{eq:opt_sim} is smooth, it is legitimate to apply Envelop Theorem. By Envelop Theorem, one can check that the derivative of $f_i$ with respect to $x_{ii}$ is equal to
    \begin{align*}
        \frac{\partial f_{i}}{\partial x_{ii}}&=\frac{\sqrt\t_i}{4v_i(b_i-x_{ii})^{\frac{3}{2}}}\left(\erf\Big(\sqrt{\pi\frac{b_i-x_{ii}}{\t_i}}R_i^*(x_{ii})\Big)-\frac{2}{\sqrt\pi}\sqrt{\pi\frac{b_i-x_{ii}}{\t_i}}R_i^*(x_{ii})e^{-\frac{b_i-x_{ii}}{\theta_i}\pi (R_i^*(x_{ii}))^2 }\right)\\
        &\quad+\frac{b_i^{\frac{2}{3}}\left(\frac{2}{\sqrt\pi}\b^*_{i}(x_{ii})e^{-(\b^*_{i}(x_{ii}))^2}-3\erf\big(\b_i^*(x_{ii})\big)+\frac{4}{\sqrt\pi}\b^*_{i}(x_{ii})\right)}{(4v_i)^{\frac{2}{3}}x^2_{ii}\left(\erf\big(\b_i^*(x_{ii})\big)-\frac{2}{\sqrt\pi}\b^*_{i}(x_{ii})e^{-(\b^*_{i}(x_{ii}))^2}\right)^{\frac{1}{3}}}
    \end{align*}
    where
    \[\b_i^*(x_{ii})=\sqrt{-\log(1-\frac{x_{ii}}{b_i}e^{-\frac{b_i-x_{ii}}{\theta_i}\pi (R^*_{i}(x_{ii}))^2}) },\]
   and $R_i^*(x_{ii})$ is the unique solution to \eqref{eq:opt_sol}. One can check that this derivative is strictly positive when $x_{ii}\in(0,b_i)$ for all $i\in\mathcal{S}$.  
   Since $f_i(x_{ii})$ is a strictly increasing function of $x_{ii}$ when $x_{ii}\in(0,b_i)$ for all  $i\in\mathcal{S}$, the uniqueness of stationary equilibrium follows from Corollary \ref{coro:uniqueness}.
\end{proof}

\end{document}